\documentclass[11pt]{article}
\usepackage{fullpage,enumitem,amssymb,xcolor}
\setlist{itemsep=.2em,topsep=.5em,parsep=.2em} 
\usepackage{fullpage}

\usepackage[utf8]{inputenc}
\usepackage[english]{babel}
\usepackage[hidelinks]{hyperref}

\usepackage{physics,url}
\usepackage{tikz}
\usetikzlibrary{positioning,calc}
\usepackage{cleveref}

\usepackage{amsmath,amsthm,amsfonts}
\usepackage{thmtools,thm-restate}
\newtheorem{theorem}{Theorem}[section]
\newtheorem{definition}[theorem]{Definition} 
\newtheorem{corollary}[theorem]{Corollary}
\newtheorem{lemma}[theorem]{Lemma}

\usepackage{graphicx}
\usepackage{xcolor}
\newcommand{\new}[1]{{#1}}
\usepackage{eqparbox}

\newcommand{\eps}{\varepsilon}
\newcommand{\cC}{ \mathcal{C}}
\newcommand{\checking}{\mathsf{Checking}}
\newcommand{\Cc}{\mathsf{C}}
\newcommand{\vac}{\bot}

\newcommand{\Span}{\mathrm{Span}}
\newcommand{\reff}{R_\mathrm{eff}}
\newcommand{\swap}{\mathsf{Swap}}
\newcommand{\send}{\mathsf{Send}}
\newcommand{\virtualsend}{\mathsf{VirtualSend}}
\newcommand{\flip}{\mathsf{Flip}}
\newcommand{\reverse}{\mathsf{Reverse}}
\DeclareMathOperator*{\E}{\mathbb{E}}

\newcommand{\set}[1]{\{#1\}}

\newcommand{\tk}{\mathsf{tk}}
\newcommand{\n}{\mathsf{n}}
\newcommand{\m}{\mathsf{m}}
\newcommand{\mm}{\mathsf{m'}}

\newcommand{\lmax}{\mathsf{l_{max}}}

\DeclareMathOperator{\polylog}{polylog}

\newcommand{\findany}{$\mathsf{FindAny}$}
\newcommand{\findmin}{$\mathsf{FindMin}$}
\newcommand{\quantumCover}{$\mathsf{QuantumCoverConstruction}$}

\usepackage{algorithm,algorithmicx,algpseudocode}
\algrenewcommand\algorithmicindent{0.5em}

\def\elected{\mbox{\small ELECTED}}
\def\nonelected{\mbox{\small NON-ELECTED}}

\begin{document}
\title{Tight Communication Bounds for Distributed Algorithms \\ in the Quantum Routing Model
%\fred{Arxiv note: Minor modifications compared to v1, in order to provide more details for the lower bounds. In order to prove the Query Complexity to Message Complexity Reduction Lemma (Lemma 7.3), we emphasize that we consider a slight variation of the adjacency array model with stronger queries.}
}

\author{Fabien Dufoulon\thanks{School of Computing and Communications, Lancaster University, Lancaster, UK.  Email: \href{mailto: f.dufoulon@lancaster.ac.uk}{\tt f.dufoulon@lancaster.ac.uk}.}
\and Frédéric Magniez\thanks{Université Paris Cité, CNRS, IRIF, Paris, France. Email:\href{mailto:frederic.magniez@irif.fr}{\tt frederic.magniez@irif.fr}. Research supported in part by the French PEPR integrated project EPiQ (ANR-22-PETQ-0007).}
\and Gopal Pandurangan\thanks{Department of Computer Science, University of Houston, Houston, TX 77204, USA. Email: \href{mailto: gopal@cs.uh.edu}{\tt gopal@cs.uh.edu}. Supported in part by ARO Grant W911NF-231-0191 and NSF grant CCF-2402837.}}
\date{}

\maketitle 

\begin{abstract}
We present new distributed quantum algorithms for fundamental distributed computing problems, namely, leader election, broadcast, Minimum Spanning Tree (MST), and Breadth-First Search (BFS) tree, in \emph{arbitrary} networks. These algorithms are (essentially) \emph{optimal} with respect to their communication (message) complexity in the {\em quantum routing model} introduced in Dufoulon, Magniez, and Pandurangan [PODC 2025]. 
The message complexity of our algorithms is $\tilde{O}(n)$ ($\tilde{O}$ hides a $\mathrm{polylog}\, n$ factor) for leader election, broadcast, and MST, and $\tilde{O}(\sqrt{mn})$  for BFS ($n$ and $m$ are the number of nodes and edges of the network, respectively). These message bounds are nearly \emph{tight} in the quantum routing model since we show almost matching corresponding \emph{quantum message lower bounds}. Our results significantly improve on the prior work of Dufoulon, Magniez, and Pandurangan [PODC 2025], who presented distributed quantum algorithms under the same model that had a message complexity of $\tilde{O}(\sqrt{mn})$ for leader election. 

Our algorithms demonstrate the significant communication advantage that quantum routing has over classical in distributed computing, since $\Omega(m)$  --- which can be as high as $\Omega(n^2)$ ---  is a well-established \emph{classical}  message lower bound for leader election, broadcast, MST, and BFS that applies even to randomized Monte-Carlo algorithms [Kutten, Pandurangan, Peleg, Robinson, and Trehan, JACM 2015]. Thus, our quantum algorithms can, in general, give a quadratic advantage in the communication cost for these fundamental problems.

A main technical tool we use to design our distributed algorithms is \emph{quantum walks based on electric networks}. We posit a framework for using quantum walks in the distributed setting to design communication-efficient distributed quantum algorithms. Our framework can be used as a ``black box'' to significantly reduce communication costs and may be of independent interest. Additionally, our lower-bound technique for establishing distributed quantum message lower bounds can also be applied to other problems.
\end{abstract}
%\maketitle

\thispagestyle{empty}
\clearpage

\tableofcontents
\thispagestyle{empty}
\clearpage
\setcounter{page}{1}

\section{Introduction and Overview}
\label{sec:intro}

The {\em communication (or message) complexity} is a critical measure of a distributed algorithm, which is defined by the {\em total number of messages} exchanged by the algorithm.  Minimizing the communication cost in distributed systems is crucial since it directly relates to energy consumption, network bandwidth usage, and overall running time and hence useful for many applications, see e.g., \cite{Woodruff,dongarra,message-byzantine,saia, AMP18}. 
There has been extensive work  
to \emph{optimize} message complexity even at the cost of increased running time or approximate quality of solution (see e.g., \cite{stoc17mst,tcssync,congested-podc2015,podc2021,itcs2024}). However, for many fundamental distributed computing problems  there are \emph{tight} lower bounds on the message complexity. In this paper, we show that one can overcome the \emph{classical} message lower bounds by taking advantage of the power of \emph{quantum} communication and routing.

\subsection{Context}

\noindent {\bf Broadcast.}
Consider the  {\em broadcast} problem, one of the most fundamental problems in distributed computing. We are given
an {\em arbitrary} network $G =(V,E)$ with $n$ nodes and $m$ edges, and the goal is to broadcast a data item initially residing
in some source node $s \in V $ to all nodes in the network.
The standard flooding algorithm, where each node (only) forwards the data item to all its neighbors when it first receives it, clearly requires $\Theta(m)$  messages. A natural question is whether one can accomplish broadcast using a significantly smaller number of messages.  (Note that $\Omega(n)$ is a trivial lower bound
on the message complexity of broadcast.)
It turns out that this is not possible, since it was established that  $\Omega(m)$ is a {\em classical lower bound} on the message 
complexity of broadcast for any distributed algorithm, even for {\em randomized Monte Carlo algorithms} (with constant success probability) and even if one wants the data item to reach only a majority of nodes~\cite{Kutten_2015_JACM}.
In fact, this lower bound shows something stronger: at least {\em a constant}  fraction of the $m$ edges of the network have to be used by any broadcast algorithm. This also shows that the broadcast bound holds regardless of the size of
a message.\footnote{In this paper, we assume messages of small size (typically $O(\log n)$ bits or qubits) as in the standard CONGEST model --- cf. \Cref{sec:model}.} Another important point to note is that the work of~\cite{Kutten_2015_JACM}  shows that the $\Omega(m)$ message lower bound holds  applies to any $m$ and $n$, i.e., for every  $m$, there exists graphs with  $\Theta(n)$ nodes and $\Theta(m)$ edges ($m= O(n^2))$, where the lower bound holds. Furthermore, the message lower bound applies {\em regardless} of the run time --- measured as the number of rounds or the {\em round complexity} --- of the distributed algorithm.
Thus, the above (classical) lower bound implies that, in general, since $m$ can be as high as $\Theta(n^2)$ (in dense graphs),
the broadcast communication cost is quadratically higher than $\Theta(n)$, the trivial lower bound.

\smallskip

\noindent {\bf Spanning Tree, MST, and BFS problems.}
The broadcast problem turns out to be equivalent in message complexity to the problem of constructing a {\em spanning tree} in a network, another fundamental distributed computing problem. Indeed, {\em once} a spanning tree is constructed, broadcast can be accomplished by sending the data item using (only) the edges of the spanning tree; clearly, this takes $O(n)$ messages. On the other hand, a broadcast algorithm can be used to construct a spanning tree without any extra messages since a node designates as its parent a node
from which it receives the item for the first time 
(if it receives from more than one at the same time, 
then it designates a unique one).  Thus, the $\Omega(m)$ message lower bound for broadcast also implies the same lower bound for spanning tree construction, and applies regardless of the round complexity and for every $m$.

The $\Omega(m)$ spanning tree lower bound also obviously applies to the construction of spanning trees with additional properties such as minimum spanning tree (MST) and the Breadth-First Search (BFS) tree, two other fundamental graph problems that have been extensively studied in distributed computing (cf.  \Cref{subsec:additionalRelatedWork}). This lower bound is tight, since there exists  $O(m)$ message algorithms for both MST and BFS \cite{Kutten_2015_JACM}.

\smallskip
\noindent {\bf Leader Election.}
The classical $\Omega(m)$ message lower bound also applies to {\em leader election}, one of the central problems in distributed computing studied
intensively for over five decades (see e.g., \cite{Lynch_1996_Book,peleg,Kutten_2015_JACM} and the references therein). Informally, leader election is the process by which nodes in a distributed network elect a single leader. Specifically, exactly one node must output the decision that it is the leader, and all other nodes must decide that they are non-leaders. The non-leader nodes need not be aware of the leader's identity.
This {\em implicit} version of leader election is standard (e.g., cf. \cite{Lynch_1996_Book}) and has been studied extensively.\footnote{In the {\em explicit} variant,
 every node must also know the identity of the unique leader.}  The $\Omega(m)$ message lower bound also applies (even) to implicit leader election and even for {\em randomized Monte Carlo algorithms} with constant success probability \cite{Kutten_2015_JACM} and {\em regardless of the round complexity of the algorithm}. Also, the bound applies to any $m$, i.e.,  given any $n$ and $m$, there exists a graph with $\Theta(n)$ nodes and $\Theta(m) = O(n^2)$ edges,
where the lower bound holds. The (implicit) leader election lower bound, in a sense, is more fundamental, since it implies the lower bounds for other problems such as broadcast, spanning tree construction, MST, and BFS \cite{Kutten_2015_JACM}.

The  $\Omega(m)$ message lower bound for all the above problems holds for graphs of {\em diameter greater than or equal to three} \cite{Kutten_2015_JACM}. However,  there are better bounds for graphs of diameter one (i.e., complete graphs) and diameter two:  in particular, for {\em leader election}, there are {\em tight} classical (randomized) message bounds of $\tilde{\Theta}(\sqrt{n})$ \cite{KPPRT15} and $\tilde{\Theta}(n)$\cite{ChatterjeePR20} for diameter one and diameter two respectively.\footnote{$\tilde{O}$ hides a $O(\polylog n)$ factor and $\tilde{\Omega}$ hides a $O(1/\polylog n)$  factor.} For spanning tree construction and broadcast, for complete and diameter two networks, $\tilde{\Theta}(n)$ is a tight bound \cite{congested-podc2015}, and for MST, $\Theta(n^2)$ is a tight bound \cite{korach,congested-podc2015}. 

\smallskip

\noindent {\bf Improving the $\Omega(m)$ classical lower bound.}
The $\Omega(m)$ classical message lower bound for all the above fundamental distributed computing problems in arbitrary networks raises the following fundamental question:
\begin{center}
    \fbox
    {
        \begin{minipage}{35em}
        Can we solve broadcast and other fundamental problems, such as spanning tree construction, MST, and BFS, with significantly less communication cost  (compared to $m$, the number of edges) in {\em arbitrary networks}   using the power of {\em quantum communication}?
        \end{minipage}
    }
\end{center}

\smallskip

\noindent {\bf Quantum communication advantage.}
The recent work of  \cite{DMP25} gave a partial answer to the above question for the problem of {\em leader election}.
This work posited a {\em quantum routing model}  of communication and presented distributed algorithms that are highly {\em communication-efficient}, in particular, those
that have {\em sublinear} in $m$ (or, even $n$) message complexity. Informally (see Section \ref{sec:model} for details), in the quantum routing model, a node can send a single message in {\em quantum superposition} to all its neighbors. Using this model, they give a quantum leader election algorithm in {\em complete networks} that, with high probability\footnote{Throughout, ``with high probability" means with probability at least $1-1/n^c$ for some constant $c$.}, elects a leader and has message complexity 
    $\tilde{O}(n^{1/3})$, beating the tight $\tilde{\Omega}(\sqrt{n})$ classical bound~\cite{KPPRT15,AMP18}. 
    For \emph{diameter-2 networks}, they present a quantum algorithm with message complexity $\tilde{O}(n^{2/3})$, beating the tight $\tilde{\Omega}(n)$ classical bound~\cite{ChatterjeePR20}.
For \emph{arbitrary networks} (which is the focus of this paper), they show that  leader election can be accomplished with message complexity 
    $\tilde{O}(\sqrt{mn})$. Note that this bound can be much smaller than the tight classical bound of $\Omega(m)$~\cite{Kutten_2015_JACM}, especially in dense graphs ($m = \Theta(n^2)$), where it gives a message complexity
    of $\tilde{O}(n^{1.5})$ messages.  
    However, they do not establish any nontrivial lower bounds and leave open the possibility that one can obtain stronger bounds in the quantum routing model.

\subsection{Classical Distributed Computing Model} 
\label{sec:model}
We first formally describe a standard distributed computing model in the classical setting, namely the synchronous CONGEST message-passing model (e.g., see 
\cite{peleg}). The quantum version of the model will be described in \Cref{sec:qmodel}.

We consider an arbitrary network  
represented as an undirected connected graph $G=(V,E)$ with $n=|V|$ nodes and $m=|E|$ edges. We will also denote by $D$ the network diameter.
Each node 
runs an instance of the same distributed algorithm.
The computation advances in {\em synchronous rounds} where, in every round, nodes
can send messages, receive messages that were sent in the same round by
neighbors in $G$,
and perform some local computation.
In the CONGEST model~\cite{peleg}, a node can send, in each round, at most one message of size $O(\log n)$ bits per edge.

All processors have access to {\em private} unbiased random bits. 
Also, we {\em do not} assume nodes have identities (using private randomness, nodes can generate
unique identifiers with high probability). 
Finally, throughout the paper, we assume that all nodes know a polynomial upper bound on $n$.

Messages are the only means of communication; in particular, nodes
cannot access the coin flips of other nodes, and do not share any memory.
Throughout this paper, we assume that all nodes are awake initially and  simultaneously start executing the algorithm. 

 We assume that nodes, initially, have only {\em local knowledge}, i.e., knowledge only of themselves; in other words, we assume the {\em clean network model} --- also called the {\em KT0 model} \cite{peleg}, which is standard and commonly used.\footnote{In Section \ref{subsec:additionalRelatedWork}, we discuss another well-studied model, namely the KT1 model, where nodes have distinct identities (IDs) and each node is assumed to know the identities of its neighbors as well. Clearly, this model does not apply to settings in which nodes lack IDs, which comes under the KT0 model.\label{ft:kt1}} Each node $v$ has $\deg(v)$ ports that it can use to communicate respectively with its $\deg(v)$ neighbors; each  port $p=(v,u)$ of $v$ is connected
 (exclusively) to a port $p'=(u,v)$ of its neighbor $u$. Finally, we denote by $N(v)$ the neighbors of $v$ in $G$.
 (Sometimes we refer to port $p$ as an integer, sometimes as $u\to v$.)

\subsection{Quantum Distributed Computing  Model}
\label{sec:qmodel}

The quantum version is analogous to the classical version described above in many respects,
while differing mainly in how messages can be sent.
In particular, we assume the {\em quantum CONGEST} model where each message consists
of $O(\log n)$ qubits \cite{ElkinKNP14}.  However, a node can send a message to {\em its neighbors} in {\em quantum superposition} as assumed in the model of quantum routing and quantum message complexity introduced in~\cite{DMP25} (and also subsequently used in \cite{legall2025, robinson2026}) that we present below. Our work takes a theoretical complexity perspective on that model. As for
a more detailed discussion on its physical feasibility, we refer to Appendix \ref{app:formal}.

\smallskip
\noindent {\bf Quantum routing model.}
We first recall the notion of quantum (or coherent) message routing introduced in~\cite{DMP25}.
This model allows a node to select {\em quantumly} which neighbors it communicates to --- that
is, its choice is in a {\em quantum superposition} --- instead of selecting via a random distribution as in the classical setting.  Of course, the message itself may consist of quantum bits.
Such a selection or control is typically referred to as \emph{coherent control} in quantum physics. 

 The quantum routing model can be viewed 
 as importing into distributed computing both the notion of quantum query complexity (see~\cite{hamoudi25} for a recent survey), where an oracle function can be queried in a superposition of inputs for the cost of a {\em single} query, and the notion of superpositions of trajectories~\cite{ck19}, which was originally defined and studied in the context of quantum Shannon theory. This work considered the case of superpositions of quantum channels~\cite{ck19}, in which a particle is sent in a coherent superposition of two or more transmission links, analogous to Young's double-slit experiment, in which a single particle can be in a superposition of two trajectories.
This has been further modeled by the notion of routed quantum circuits~\cite{vkb21}. Similar notions of accesses in superposition have been considered for quantum random access memory (QRAM) and quantum random access gates (QRAG). This has been considered for quantum algorithms~\cite{Amb07} and quantum programming languages~\cite{ACCRV23}. 

Since we aim for a message complexity that is sublinear in $m$ (the number of edges) it is crucial that we allow the possibility of controlling the message recipient quantumly. Thus, in the quantum routing model, a node may send a message to all (or a subset of its neighbors) in quantum superposition, and this is counted as {\em one} message. To illustrate the improvement in message complexity we can obtain by quantumly controlling messages, we consider a star graph on $n$ nodes where the central node has to locate a (unique) leaf node with a specific property. This search problem can be solved by using $O(\sqrt{n})$ quantum message complexity (with constant success probability) using the distributed Grover search algorithm (see \cite{DMP25} for details). This is in contrast to the $\Omega(n)$ classical lower bound. This is possible quantumly, since the central node can send a message in {\em superposition to all other nodes}  ---  which is counted as one message in the quantum routing model. 

We exploit a similar phenomenon to design communication-efficient distributed {\em quantum walks}. Intuitively, these quantum walks can detect the presence of a specific node quadratically more efficiently than classical walks (see \Cref{subsec:distQuantWalksOverview}). 
Just as in random walks, the quantum walk's token proceeds one step per (quantum-routed) message, but in superposition to all of its neighbors and so on.
Then, just as in classical random walks, the $t$ steps of the quantum walk incur $t$ messages. 
\smallskip

\noindent {\bf A formal description.}
Similar to the classical setting (see \Cref{sec:model}), each node $u$ has $\deg(u)$ ports, one for each node $v$ connected to $u$. For each such port, $u$ holds two registers ${u\to v}$ and ${u\gets v}$ --- respectively called \emph{emission and reception registers}. With some abuse of notations, we denote a port $p$ by the pair $p=(u,v)$, meaning that this is a port based in $u$ for communicating with $v$. 
Next, we formally describe the communication in the quantum routing model. For any node $u$, each communication consists of (1) message preparation, and (2) message emission. 

For message preparation, $u$ first sets up some superposition $\ket{\psi}=\sum_{p=(u,v)} \alpha_p \ket{p}$ of ports (and therefore of nodes) as well as some message $\ket{m}$ (where $\ket{m}$ could itself be port dependent). Note that since some $\alpha_p$ could be $0$, this also captures $u$ selecting a subset of ports, and furthermore, when only one port $p$ satisfies $\alpha_p\neq 0$, the selection is deterministic. Then, $u$ prepares the message (using unitaries called control-swap) as follows, where all explicit registers are located in $u$:
$$\Big(\ket{m}\otimes \Big(\sum_{p=(u,v)} \alpha_p \ket{p}\Big)\Big)_u\bigotimes_{p=(u,w)} \ket{\vac}_{u\to w}  
\mapsto  \sum_{p=(u,v)} \alpha_p \Big( \ket{\vac} \otimes \ket{p} \Big)_u \otimes\ket{m}_{u\to v}\bigotimes_{p=(u,w\neq v)} \ket{\vac}_{u\to w}.$$

For message emission, $u$ applies the $\send_{u}$ operator to the above obtained (local) quantum state. That operator models the message emission from $u$ to all of its neighbors, and is a tensor product of the operators for message emission over the ports of $u$. More precisely, for any port $p=(u,v)$ of $u$, the corresponding operator simply swaps ${u\to v}$ in $u$ with register ${v\gets u}$ in $v$: $\send_{u\to v} (\ket{m}_{u\to v}\ket{\vac}_{v\gets u})= \ket{\vac}_{u\to v}\ket{m}_{v\gets u}$. As for $\send_u$, it is simply $\send_u = \bigotimes_{p=(u,v)}  \send_{u\to v}$.
Then, the message emission of $u$ (where for the communication registers, we explicit only the emission registers of $u$, and the reception registers of $u$'s neighbors) results in 
$$ \send_{u}\left(\sum_{p=(u,v)} \alpha_p \Big( \ket{\vac} \otimes \ket{p} \Big)_u \otimes\ket{m}_{u\to v}\bigotimes_{p=(u,w\neq v)} \ket{\vac}_{u\to w}  \bigotimes_{p=(z,u)} \ket{\vac}_{z\gets  u}\right) $$
which turns out to be
$$    \sum_{p=(u,v)} \alpha_p \Big( \ket{\vac} \otimes \ket{p} \Big)_u 
    \bigotimes_{p=(u,w)} \ket{\vac}_{u\to w} \otimes \ket{m}_{v\gets u} \bigotimes_{p=(z\neq v,u)} \ket{\vac}_{z\gets  u}.$$
In this example, one can see that, in each term of the superposition, the message $m$ has been sent to a single node $v$.

Finally, note that the communication over the network is done by having all nodes prepare their message (if they have any), then apply the global operator for the emission of all messages sent across the network, which is simply $\send=\bigotimes_{u}\send_u$.

\smallskip

\noindent {\bf Quantum message complexity.}
We continue with the notion of quantum message complexity.

At any time, the system's configuration can be in a superposition of all possible behaviors of deterministic configurations. The transition from one configuration to another is done according to a unitary transformation, made of two steps: (1) Perform the same local unitary to each node on their local data, local memory, and reception/emission registers; (2) Send non-empty messages prepared in step (1), where, implicitly, a non-empty register selects a node to send a message.

The \emph{message complexity} $M$ of a distributed algorithm is simply the sum of the message complexities during each of its rounds.
Deterministically, a round of communication has \emph{message complexity} $C$ when the network carries at most $C$ messages of $O(\log(n))$ quantum bits in that round.
Quantumly, we extend this notion to a superposition of configurations.
A round of \emph{quantum} communication has \emph{message complexity} $M$ when the global state of the network 
is in a superposition of (deterministic) configurations with message complexity at most $M$.

\subsection{Problems}\label{sec:pbs}
For all the problems below, we assume that we are given an {\em arbitrary} (undirected) connected network $G=(V,E)$ with $n = |V|$ nodes and $m = |E|$ edges. Furthermore, for the MST problem, we assume that edges have weights bounded by a polynomial in $n$. 
The goal is to design distributed algorithms with as small a {\em message complexity} as possible. 

\smallskip
 \noindent {\bf Leader Election.}
Every node $u$ has a special variable $\texttt{status}_u$ that it can set
to a value in $\{\bot, \nonelected, \elected \}$; initially we assume
$\texttt{status}_u = \bot$.
An \emph{algorithm $A$ solves leader election in $T$ rounds}
if, from round $T$ onward, exactly one node has its status set to $\elected$ 
while all other nodes have theirs set to $\nonelected$. This is the requirement 
for standard (implicit) leader election. 

\smallskip
\noindent {\bf Broadcast.} Given a distinguished source node $s \in V$ that has a data item (say of size $O(\log n)$ bits),
the goal is to disseminate the item to all nodes in the network.

\smallskip
\noindent{\bf Spanning Tree (ST) and Minimum Spanning Tree (MST).} Given a connected graph, the goal is for each node to output (i.e., should locally know) the edges of a spanning tree {\em incident on itself}. For MST, we assume arbitrary edge weights (weights are at most a polynomial in $n$) and,  
 as is standard, we assume that they are unique (this makes
the MST unique) \cite{eatcsmst}. The goal is for each node to output the MST edges incident on itself. 

\smallskip
\noindent {\bf Breadth-First Spanning Tree (BFS).} Given a distinguished node $r \in V$, the goal is for each node to output the edges of a BFS tree rooted at $r$ that are {\em incident on itself}.

\subsection{Our Contributions}
\label{sec:contributions}

We present two sets of results: distributed quantum algorithms and quantum message lower bounds.

\smallskip

\noindent {\bf Distributed quantum algorithms.} 
We present new communication-optimal distributed quantum algorithms  in the quantum routing model for all the fundamental distributed computing problems discussed earlier, namely,  {\em leader election}, {\em broadcast}, {\em Minimum Spanning Tree (MST)}, and {\em Breadth-First Search (BFS) tree}, in {\em arbitrary} networks.
For leader election and MST, we give distributed 
quantum algorithms with message complexity of $\tilde{O}(n)$ (cf. \Cref{sec:LE}).
This also implies the same bounds for broadcast and spanning tree construction, since an MST is also an ST and can be used to do broadcast using an additional $O(n)$ messages. 
For BFS tree construction, we present a distributed quantum algorithm with message complexity $\tilde{O}(\sqrt{mn})$ (cf. \Cref{sec:BFS}).

Our results demonstrate the significant communication advantage that quantum routing has over classical in distributed computing since broadcasting, leader election and MST construction can be accomplished by sending only $\tilde{O}(n)$ messages and BFS in $\tilde{O}(n^{1.5})$ messages. This  is impossible, in general, in the classical setting, since $\Omega(m)$  --- which can be as high as $\Omega(n^2)$ ---  is a well-established
{\em classical}  message lower bound for leader election, broadcast, MST, and BFS  
that applies even to randomized Monte-Carlo algorithms \cite{Kutten_2015_JACM}. Thus, in general,  for leader election, broadcast, and MST, quantum routing capabilities can give a {\em quadratic} advantage in message complexity over classical, and for BFS, a $\sqrt{n}$ factor advantage.

Our results also improve over the work of Dufoulon, Magniez, and Pandurangan \cite{DMP25}, who presented distributed quantum algorithms with message complexity of $\tilde{O}(\sqrt{mn})$ for leader election and MST in arbitrary networks. For BFS, no sublinear-in-$m$ message complexity quantum algorithm was known, prior to our work, for arbitrary networks.

While we do not focus on optimizing the {\em round complexity} of our algorithms, our algorithms for leader election, MST and broadcast run in $\tilde{O}(n)$ rounds, and our algorithm for BFS runs in $\tilde{O}(D \sqrt{n})$ round ($D$ is the network diameter). 
These can be higher than the $\Omega(D)$ lower bound on the round complexity for all these problems except for MST, which has a $\tilde{\Omega}(D+\sqrt{n})$ lower bound (these respective round lower bounds apply for both classical
and quantum \cite{Kutten_2015_JACM,ElkinKNP14,GKM09}). However, in general, our $\tilde{O}(n)$ round complexity is optimal in graphs of diameter $\Theta(n)$. It is open whether we can design distributed quantum algorithms that are simultaneously optimal
with respect to  both message and round complexity  (of the respective problems) and we conjecture that this is not possible in the quantum setting, unlike in the classical setting where such algorithms are known for all the above problems \cite{Kutten_2015_JACM, stoc17mst} (cf. \Cref{sec:conclusion}).

The key technical tool we use to design our distributed algorithms for leader election and MST is {\em quantum walks} based on {\em electric networks}.
We propose a framework for utilizing quantum walks in a distributed setting to design communication-efficient distributed quantum algorithms. This can be useful for other distributed problems as well (cf. \Cref{sec:conclusion}). 

\smallskip
\noindent {\bf Quantum message lower bounds.} 
    We show that all our algorithms are indeed message-optimal (up to a $O(\polylog n)$ factor) for the respective problems by showing almost {\em matching}  quantum message lower bounds in the quantum routing model (cf. \Cref{sec:lb}). 
    To the best of our knowledge,  these are the first non-trivial {\em message} lower bounds (cf. Theorems \ref{lb:BFS} and \ref{lb:LE}) for  distributed quantum computing. 
    
    Specifically, we show two main results. First, any quantum distributed algorithm solving (even) implicit leader election on graphs of diameter at least three must send $\Omega(n)$ quantum messages. The same lower bound for MST and ST follows from this lower bound.
    Second, any quantum distributed algorithm for constructing a BFS tree must send $\Omega(\sqrt{mn})$ quantum messages. This also trivially implies the same lower bound for the Single Source Shortest Paths (SSSP) problem.

    Our technique for showing quantum message lower bounds may also be helpful for other problems, and is as follows. We first show lower bounds on the
    {\em quantum query complexity} of the graph problem of interest (say BFS or leader election) on a suitably constructed graph family in the {\em centralized setting}.
    %\fabien{Add detail here.}
    In this setting, we assume that the algorithm is given a query access to the {\em adjacency array} of the graph \cite{DHHM06}.
    %\new{(in fact, with slightly stronger queries than \cite{DHHM06})}. 
    Then we leverage a generic reduction lemma (cf. \Cref{lem:CentralizedToDistributedLowerBound}) 
    to ``lift" these lower bounds shown in the (centralized) query complexity model for an (unweighted) graph problem to the corresponding message complexity lower bounds.

\section{Technical Overview}
\label{sec:high-level}

We give the high-level idea behind our results, starting from our distributed quantum walks. Then, we describe our $\tilde{O}(n)$ communication MST (as well as leader election and broadcast) algorithm, which hinges on these distributed quantum walks. Finally, we give some intuition regarding our $\tilde{O}(\sqrt{nm})$ communication BFS algorithm as well as our quantum message complexity lower bounds. 

\subsection{Distributed Quantum Walks}
\label{subsec:distQuantWalksOverview}

We first consider a sequential setting and give some intuition regarding quantum walks that are based on the framework of {\em electric networks}. (Note that we give here a somewhat informal version of sequential quantum walks; interested readers can find the formal version in \Cref{electric}.) Then, we describe what obstacles we overcome for our distributed implementation of such quantum walks.

\smallskip
\noindent {\bf Sequential Quantum Walks and (Marked) Element Finding.} Both classical random walks, and quantum walks, can be designed and analyzed via the framework of electric networks.

In the classical setting, take a (single) random walk on some graph $G = (V,E)$ whose behavior is described by weights $(w_e)_{e \in E}$ on the edges of $G$: that is, at node $v$, the random walk moves along some edge $e \in E$ to neighbor $u$ with probability $w_{e}/w_v$ where $w_v = \sum_{e' \in E : e' \ni v} w_{e}$. Suppose that the walk starts at some node $r$, and that we are interested in the expected number of steps required to reach some subset $M \subseteq V$ of marked nodes --- that is, the \emph{hitting time} $H_{r,M}$. Then, an upper bound on this hitting time can be obtained by considering the corresponding electric network with the same graph structure as $G$, where the weight of some edge $e$ denotes the conductance (i.e., the inverse of the resistance) of that edge, and one unit flow $f$ goes from (source node) $r$ to the (sink) nodes of $M$. 
The energy dissipated by the flow is $\sum_{e \in E} f_e^2/w_e$. The minimum dissipated energy over all unit flows, which is achieved when the unit flow corresponds to electrical current and thus follows Kirchhoff's Laws, is the \emph{effective resistance} between $r$ and $M$. It was shown (see, e.g., \cite{CRRST97}) that the hitting time can be upper bounded by the effective resistance, that is, $H_{r,M}= O(RW)$ where $R$ is any upper bound on that effective resistance and $W$ is the sum of all weights in $G$.

We can quantize the above electric framework following the approach of Belovs \cite{belovs2013} (see Section \ref{subsec:additionalRelatedWork} for a more in-depth discussion), with some minor changes. To do so, we adapt the above walk as follows (see also~\cite{apers22,hari25,zur25} for similar but slightly different approaches).
First, we add a single ``artificial'' node denoted by $\overline{r}$ to $G$, connected to its only neighbor $r$ via an edge of weight $1/(C_1R)$, for some constant $C_1 \geq 1$ --- this artificial node plays the role of a (delayed) self-loop on $r$. Second, due to reversibility, the walk happens on edges --- capturing the current node on which the walk lies, but also the previous node so we can reverse the walk --- rather than on the vertices. More formally, vectors $\{\ket{u,v}, \ket{v,u} \mid \{u,v\} \in E\}$ form the computational basis of the vector space of the quantum walk.
Finally, a single step of the quantum walk --- denoted by operator $U$ --- consists in applying a diffusion operator $D$ followed by a swap operator $(-\swap)$. Intuitively, the diffusion serves as the quantum analog, when the walk is on $v \notin M$ to choosing a random incident edge $e$ with probability proportional to $w_{e}$ (and when $v \in M$, then the diffusion is the identity) whereas the swap moves the quantum walk over edge $e$ to the other endpoint (the minus sign is a technical detail that we can ignore here).

What properties hold for this quantum walk? 
First, the quantum walk operator $U$ is unitary, thus the eigenvalues of $U$ have unit modulus, that is, they are of the form $e^{i\alpha}$ where we can take, w.l.o.g., $\alpha \in (-\pi,\pi]$. (Note that as a result, quantum walks do not converge to a stationary distribution.) Furthermore, it turns out that 
the state $\ket{\sigma} = (\ket{r, \overline r} - \ket{\overline r, r})/\sqrt 2$ has a strong overlap with the 1-eigenspace of $U$ when $M \neq \emptyset$, whereas $\ket{\sigma}$ has a small overlap not only with the 1-eigenspace of $U$ but also with the span of all eigenvectors of $U$ with phase smaller than $\Theta(1/\sqrt{RW})$, when $M = \emptyset$. 

These properties of $\ket{\sigma}$ play a crucial role in how we use quantum walks. Indeed, they allow us to use quantum phase estimation on $\ket{\sigma}$ to detect the presence of marked nodes (with constant probability). More precisely, even though $\ket{\sigma}$ is not itself an eigenvector of $U$, when $M \neq \emptyset$ then the estimated phase of $\ket{\sigma}$ is 0, whereas otherwise (when $M = \emptyset$), the estimated phase is far away from 0. These two statements hold with some probability, that depends on the number of quantum walk steps executed within the phase estimation procedure. However, $\tilde{O}(\sqrt{RW})$ quantum walk steps suffice to obtain a high enough precision estimation (for our needs) on the phase of $\ket{\sigma}$.
Finally, we can combine this estimation procedure with a simple binary search over the nodes --- where only marked nodes within the considered ``node range'' still act as marked nodes in the walk --- to \emph{find a marked node}, at an additional $O(\log n)$ multiplicative factor.

\smallskip
\noindent {\bf Distributed Implementation.} The above sequential quantum walk can be naturally adapted, communication-efficiently, to the (distributed) quantum routing model. Indeed, that walk consists of a {\em single token} which moves once per step of the walk, and each such step can be implemented via a {\em single quantum-routed message} (cf. Lemma \ref{dqw}). Note that this remains the case even as the token of the quantum walk, initially located on edges incident to the root, lies in a superposition over an increasingly large support (i.e., edges that the token is on, upon measurement) as the walk takes more and more steps. 
However, adapting how such quantum walks find (marked) elements to the quantum distributed setting is non-trivial.
The first obstacle lies in the centralized quantum control required by the quantum phase estimation procedure. This centralized control can lead to a large message overhead. However, we circumvent this by restricting ourselves to a more basic primitive, that we call \emph{quantum phase detection}, which is sufficient for our purpose and easier to adapt to distributed computing. 

A second obstacle lies in the need, for our distributed implementation, to concurrently execute multiple distributed quantum walks that, at first glance, interfere with each other. Now, if these walks were defined on pairwise node-disjoint subgraphs of the communication graph $G$, then they would not interfere with each other and thus would correctly find a marked element per subgraph (given the above phase detection procedure). However, our distributed implementation executes quantum walks that satisfy weaker properties: that is, these quantum walks run on subgraphs that are edge-disjoint, except for edges that contain some walk's marked element. Nonetheless, we modify the naive, distributed quantum walk step to avoid interference between multiple such quantum walks, using the crucial property that diffusion operator is the identity on marked elements.  

\subsection{Leader Election and MST in $\tilde{O}(n)$ Quantum Communication}
\label{subsec:high-levelLE}

Consider the canonical spanning tree (ST) construction problem; it is not difficult in the distributed setting to modify the solution to ST so that it solves MST (and hence, leader election and broadcast). In the ST problem, we have to construct a spanning tree in an  arbitrary input graph. If we want to accomplish this in a communication-efficient way, then we apply the well-known Gallager-Humblet-Spira (GHS) distributed algorithm (see e.g., \cite{GHS, peleg, eatcsmst}). This can also be  thought of as a {\em tree merging algorithm} \cite{eatcsmst}. 

\smallskip

\noindent {\bf Context.}
The GHS algorithm begins with $n$ singleton clusters, one for each node. In every phase, it merges some clusters so that the number of clusters goes down by a constant factor. A cluster (say $C$) is merged with another cluster by finding an {\em outgoing edge}, i.e., an edge with one endpoint in $C$ and the other in $V\setminus C$.
Thus the key step in the GHS algorithm that is non-trivial to implement communication-efficiently is finding an outgoing edge.
Classically, this can be implemented by each node in the cluster communicating through all of its incident edges to figure out if the other endpoint lies outside the cluster. This essentially gives rise
to a $\tilde{\Theta}(m)$ messages algorithm and this is (essentially) message optimal \cite{Kutten_2015_JACM}.

Quantumly, one can do better by using {\em Grover search} (see e.g., \cite{Grover96,bbht98,DHHM06}). 
Indeed, each node can perform a Grover search to find an outgoing edge (if any exists) among its incident edges in the quantum routing model \cite{DMP25}. Plugging this quantum outgoing edge search primitive into the GHS distributed algorithm gives rise, as shown by \cite{DMP25}, to a $\tilde{O}(\sqrt{mn}) = \tilde{O}(n^{1.5})$ messages ST algorithm. 

\smallskip
\noindent {\bf Challenges.}
Solving ST using only $\tilde{O}(n)$ message complexity boils down to finding an outgoing edge, for any cluster $C$, using only $\tilde{O}(|C|)$ messages in the quantum routing framework. At first glance, this looks plausible since any cluster must have an $f = \Omega(1/|C|^2)$ fraction of outgoing edges\footnote{Indeed, if we let $E_{out}$ and $E_{in}$ be respectively the number of outgoing edges and that of edges between two nodes in $C$, then the fraction of outgoing edges is $E_{out}/(E_{in}+E_{out}) \geq E_{out}/(|C|^2+E_{out}) \geq 1/(|C|^2+1)$.}, thus a Grover search on the edges incident to the cluster need only send a message over (or query) $\Theta(\sqrt{1/f}) = \Theta(|C|)$ of the incident edges to find an outgoing edge with constant probability. 

However, such a Grover search is not easy to implement in the distributed setting. Indeed, 
if each node runs its own Grover search, then $\tilde{O}(\sqrt{mn})$ messages are sent, just as in \cite{DMP25}. On the other hand, one could run a centralized Grover search, where the root node of the cluster controls the search over the cluster's incident edges.
However, this centralized search comes at a \emph{locality cost}. In particular,
the cluster $C$ may have a large diameter (up to $\Theta(|C|)$) and thus one has to pay, \emph{for each query of the Grover search}, the message cost of communicating back and forth from the root.
As a result, this centralized Grover search could take up to $O(|C|^2)$ messages.

\smallskip
\noindent {\bf Our Approach.}
We overcome the locality cost of a centralized search over a cluster's incident edges through a novel use of
{\em quantum walks}. 
Quantum walks can be thought of as the quantum counterpart of classical random walks \cite{ADZ93}.

More precisely, to find the outgoing edge for any cluster $C$, we consider the corresponding subgraph $H$ induced by the set of edges incident to (any node in) $C$, and ``mark'' the nodes of $H$ that are outside $C$, if there exist any. Then, first note that if $H$ contains at least $|C|^2+1$ edges, finding an outgoing edge can be done simply by looking for some node in $C$ with degree at least $|C|+1$ within $H$ and checking among its incident edges for some outgoing edge. 

On the other hand, if $H$ contains less than $|C|^2$ edges, then we can search for a marked node (if one exists) by using a distributed quantum walk on $H$ that starts at the root of the cluster $C$ and uses the following edge weights: for each edge in the cluster's spanning tree, assign a weight that is equal to its distance from the root, and for each other edge in $H$, a unit weight. Then, the edge weights add up to $W = O(|C|^2)$ in total, and $R = O(\sum_{i=1}^{|C|} 1/i) = O(\log |C|)$ upper bounds the effective resistance.\footnote{Note that this non-uniform assignment of edge weights is crucial. Indeed, if edges in $H$ all have unit weights then $W = O(|C|^2)$ and $R = O(|C|)$ only. The resulting upper bound on the quantum walk steps is only $O(|C|^{3/2})$.} In which case, the quantum version of the above walk takes $\tilde{O}(\sqrt{R W}) = \tilde{O}(|C|)$ steps (cf. Section \ref{subsec:distQuantWalksOverview}), and its distributed implementation takes as many messages. 

One downside of the walk is that it only detects the presence of an outgoing edge, rather than find one. Still, we can execute a binary search over the edges incident to $C$, but excluding the edges in the cluster's spanning tree. For each reduced subgraph (i.e., reduced set of edges, and cluster's spanning tree) that we obtain, we execute a (distributed) quantum walk. This induces only an $O(\log n)$ multiplicative blow-up.

%Fabien: change here
\subsection{BFS in $\tilde{O}(\sqrt{mn})$ Quantum Communication}
\label{sec:high-levelBFS}

The traditional distributed algorithm for building a BFS tree does so via flooding the communication graph. This builds the BFS tree layer per layer, in an ``inside-out'' fashion. More precisely, the root sends messages to its neighbors, which enter the first layer of the BFS. These newly joined nodes in turn send messages to their neighbors, and among them, those neighbors not already in the BFS join the second layer, and so on. However, it is not clear how quantum routing can help significantly reduce the message complexity of this basic BFS algorithm. 

Our distributed quantum BFS algorithm achieves $\tilde{O}(\sqrt{mn})$ communication by changing the above scheme in two major ways. The core reason for theses changes lies in an \emph{efficient reduction} of the BFS algorithm to search problems over nodes' 1-hop neighborhoods. Indeed, it is well-known that a search problem over $x$ elements can be solved with a quadratic improvement in the quantum setting: Grover Search uses only $O(\sqrt x)$ quantum queries over the elements, whereas classically, the search problem needs $\Theta(x)$ queries. The quantum routing model enables the use of distributed Grover Search to perform a communication-efficient search over a node's 1-hop neighborhood~\cite{DMP25}.

Our first change is to build the BFS tree layer per layer but in an ``outside-in'' fashion. More precisely, whenever the BFS tree has been built up to some layer $i \geq 1$, then \emph{nodes outside the BFS tree send messages and find out if one of their neighbors is in layer $i$}. 
This can be done via distributed Grover Searches: each node not yet in the BFS tree executes such a search over its 1-hop neighborhood. As a result, we grow the BFS tree by one layer using $\tilde{O}(\sqrt {\deg(v)})$ messages per node $v$, where $\deg(v)$ denotes the degree of $v$, and thus $\tilde{O}(\sqrt{mn})$ messages overall (by the Cauchy-Schwarz inequality).
Unfortunately, this one change does not suffice, since the resulting message complexity grows linearly in the depth of the BFS tree. Nonetheless, we can use this algorithm to build BFS trees up to some small $\polylog n$ depth using only $\tilde{O}(\sqrt{mn})$ messages, which we crucially use in our second change.

Our second change aims to significantly reduce the number of times each node outside the BFS runs Grover Search to join the BFS, down to $O(\log n)$ times. To do so, yet compute a BFS tree correctly, we follow the scheme from \cite{DPRS25}. More precisely, we compute a sparse neighborhood cover, that is, a collection of (not necessarily disjoint) trees such that each tree has depth $O(\log n)$, each node is in at most $O(\log^2 n)$ of them, and the (closed) neighborhood of any node $v$ is entirely contained within one of these trees. Given such a cover, we can ensure that all nodes learn whenever they are close to the current outermost layer of the computed tree, using $\tilde{O}(n)$ messages only. As for computing such a cover, this can be done by plugging in the algorithm described above, that uses $\tilde{O}(\sqrt{mn})$ messages to build low-depth (i.e., $\polylog n$ depth) BFS trees, into the sparse neighborhood cover construction due to Elkin \cite{Elkin06}.

\subsection{Quantum Communication Lower Bounds}
\label{sec:high-levelLB}

Our quantum communication lower bounds leverage a simple, generic reduction from the quantum query complexity of any graph problem (see e.g., \cite{DHHM06}) on some unweighted graph $G$, in a centralized setting where one initially knows the degree of all nodes in $G$ and accesses $G$ via adjacency arrays \new{--- which amounts, essentially, to the adjacency array model \cite{DHHM06} but with slightly stronger queries ---} to the quantum communication complexity of the same problem in a distributed setting where $G$ is the communication graph. In addition to the reduction, our lower bounds rely on two separate quantum query complexity lower bounds. 
%\new{(Note that our lower bounds also hold for the original adjacency array model.)}\fred{Maybe suppress this note.} 
The first quantum query complexity lower bound shows that graph connectivity on graphs of diameter 3 and above requires $\Omega(n)$ quantum queries, and the second one that computing a BFS requires $\Omega(\sqrt{mn})$ quantum queries. 
To obtain these two lower bounds, we use the adversary method with non-negative weights \cite{Ambainis02}. Note that since we are initially given the nodes' degrees, our lower bound graphs must ensure that the nodes' degrees do not provide any information.

For graph connectivity, we leverage lower bound graphs typically used for (classical) message complexity lower bounds of distributed leader election. More precisely, we consider two families of graphs: the bad instances, which are made up of two disconnected cliques of size $n/2$ each, and the good instances, which are connected diameter-3 $n$-node graphs obtained from bad instances via a small change on the adjacency arrays: taking one edge per clique, say $\{a,b\}$ and $\{c,d\}$, and crossing them, i.e., replacing the two edges by \emph{bridge edges} $\{a,c\}$ and $\{b,d\}$. Intuitively, a good and bad instance being related via such a small, hard to detect change --- the two bridge edges are ``hidden'' within $\Theta(n^2)$ edges --- in the adjacency arrays makes it hard for the algorithm to tell the two instances apart with few queries. Whereas classically, this implies that connectivity requires $\Omega(n^2)$ queries, in the quantum setting, these lower bound graphs combined with the adversary method imply that $\Omega(n)$ queries are needed.

For BFS, we construct a different family of lower bound graphs. For any fixed node size $n$ and degree $d \leq n-1$, the lower bound $n$-node $nd$-edge graphs essentially have the same node and edge structure (see Figure \ref{fig:bfs}), but differ in how edges are locally ordered. For example, if we take a graph $G$ and permute two edges incident to some node $v$ within the node's local ordering, we generate a different lower bound graph $G'$, in the sense that $G$ and $G'$ differ in the adjacency array of $v$. Then, a natural definition of what it means to solve BFS in this setting is the following: compute a set of edges that defines a BFS tree rooted in some input node $r$, as well as each such edge's position within the child endpoint's local ordering. Our lower bounds' graph structure is simple but ensures that we have a matching $M$ of $n$ edges, ``hidden'' among the total $m=nd$ edges, such that any BFS computed in that lower bound graph must \emph{contain (and thus detect) all edges of $M$}. Via the adversary method, we show that this leads to an $\Omega(\sqrt {nm})$ quantum queries for computing BFS.

\section{Additional Related Work and Comparison}
\label{subsec:additionalRelatedWork}

\noindent {\bf Classical distributed setting.}
The problems addressed in this paper, namely leader election, broadcast,  MST, and BFS construction, have been studied extensively in  {\em classical} distributed computing for several decades, both with respect to optimizing both message and round complexity, the two fundamental performance measures of distributed algorithms. For example, we refer to \cite{Lynch_1996_Book, peleg, KPPRT15, Kutten_2015_JACM, eatcsmst, stoc17mst} and the references therein. To summarize, the above results show that all of the above problems have a {\em tight} $\Theta(m)$ message complexity. As far as round complexity is concerned, all of the above problems, except MST, have a  tight  $\Theta(D)$ round complexity, where $D$ is the network diameter; MST, on the other hand, has a tight
$\tilde{\Theta}(D+ \sqrt{n})$ round complexity.  There also exist {\em singularly-optimal} algorithms that  simultaneously have optimal message and round complexity for all of these problems \cite{Kutten_2015_JACM, stoc17mst, elkin2017}.

\smallskip

\noindent {\bf The KT1 variant.}
It is worth noting that the tight $\Theta(m)$ message complexity applies to the clean network model (\Cref{sec:model}), also known as the KT0 model, which is assumed here and widely used. 
As mentioned in \Cref{sec:model} (see Footnote \ref{ft:kt1}), the above problems have also been studied in an alternate model, called the {\em KT1 model}, which is stronger in the sense that nodes are assumed to have unique IDs and   have knowledge of the  IDs of their neighbors (in addition to their own IDs).
This assumption can entail a significant amount of initial knowledge; in particular, a node potentially needs to know $\Theta(n)$ information if its degree is $\Theta(n)$.
If KT1 knowledge is assumed, then all the above problems (leader election, broadcast and MST), {\em with the notable exception of BFS},
can be solved (classically) in $\tilde{O}(n)$ messages (using randomization)\cite{KKT}, but taking $\tilde{O}(n)$ rounds (in the CONGEST model). It is remarkable to point out that these are the same bounds achieved by our quantum algorithms as well, but  {\em without any knowledge of the neighbor IDs}. 

For BFS, obtaining a $o(m)$ message complexity in KT1 CONGEST, in particular obtaining a $\tilde{O}(n)$ complexity, is {\em open in the classical setting}.\footnote{In KT1 LOCAL, where there is no bound on the message sizes, all the problems considered here can be solved singularly optimally in $\tilde{O}(n)$ messages and $\tilde{O}(D)$ rounds \cite{DPRS25}.} In contrast, we show that in the quantum setting, one can achieve $\tilde{O}(\sqrt{mn}) = \tilde{O}(n^{1.5})$ message complexity for BFS {\em without any initial knowledge}. As mentioned earlier,  a $o(m)$
message complexity result is not known (even) in
the classical KT1 CONGEST model. 

\smallskip

\noindent {\bf Quantum distributed setting.}
While the focus of this paper is on optimizing the message complexity, there have been some previous works in distributed quantum algorithms that focus on the {\em round} complexity.
One of the earliest studies of round complexity in quantum distributed computing in the synchronous message-passing model 
was  in the LOCAL model, where there is no bound on the message size per round \cite{DP08,GKM09}.  Separations between the computational powers (with respect to round complexity) of the classical and quantum versions of the model have been reported for some non ``natural'' problems~\cite{GKM09,GallNR19,Balliu25}, but other papers have also reported limited improvement for other problems (e.g., approximate graph coloring~\cite{Coiteux-RoyDGKG24}). 

In the quantum setting, it was shown in~\cite{ElkinKNP14}  that the quantum CONGEST model 
is not more powerful than the classical CONGEST model for many important graph optimization problems including MST. Nonetheless, it was later shown in~\cite{LM18} that computing a network's diameter can be solved faster in the quantum setting. Since then, other quantum speed-ups have been discovered, in particular, for subgraph detection~\cite{CFGLO22,ApeldoornV22,FraigniaudLMT24}.

Regarding work on optimizing the  {\em message} complexity of distributed quantum algorithms, the development is more recent \cite{roget2025, DMP25,legall2025,robinson2026}.
This can be understood by several recent developments of sublinear message complexity randomized (and thus classical) protocols, but also by several physics experiments demonstrating that quantum routing is physically achievable (cf. \Cref{app:formal} for more discussion on this).

A notable work on improving the message complexity of distributed quantum algorithms is that of \cite{icalp24}. This work improves the message complexity of fault-tolerant (crash) consensus in complete networks, but assumes (as do all previously mentioned works that focus on round complexity) that message exchanges (classical or quantum bits) are routed classically, and not quantumly, as is done in the quantum routing model. However, the $\tilde{O}(n)$ message complexity that they achieve does not substantially improve over the classical message complexity of fault-tolerant consensus in complete networks. This is in sharp contrast to the results we give here.

\smallskip

\noindent {\bf Quantum walks.}
We use a framework of electric networks for designing efficient quantum algorithms using quantum walks ~\cite{belovs2013,BelovsCJKM13}. 
The concept of quantum walks was introduced for path graphs~\cite{ADZ93} and regular graphs~\cite{Watrous01}. Later on, the quantum analogue of any Markov chain on any graph was defined~\cite{Szegedy04}, with a specific interest in random walks on weighted graphs, that is, reversible Markov chains. Recent works consider much more elaborate walks that we do not consider here, such as high-dimensional quantum walks~\cite{JefferyZ25}.

Inspired by the classical notion of hitting times, quantum walks were used to search for some specific marked elements in a graph by doing a local search guided by the walk. They can be interpreted as a generalization of the Grover search algorithm. Where a Grover search corresponds to a search without any local constraints, that is a walk on the complete graph, a quantum walk can, for instance, search highly constrained spatial structures such as the boolean hypercube~\cite{SKW03}, the 2D-grid~\cite{AKR05} and even ad-hoc structures such as the welded tree~\cite{CCDFGS03,LLL24}. The search on the welded tree originally exploited the notion of continuous-time quantum walks~\cite{FG98} that we do not consider in this work.
Finally, a generic way to use quantum walks in order to speed up Grover search has been introduced in~\cite{Szegedy04,mnrs}, leading to general frameworks with several applications. However, one limitation of those models is that they must start with an initial superposition corresponding to the quantum analogue of a stationary distribution, which would require some additional communication costs. 

Later on, the full potential of quantum for local search was unleashed by considering the simpler problem of detecting the presence of any marked elements, instead of finding one of them. This subtlety allows a much more general framework such as the one we are going to present below. In particular, it becomes possible to initiate the quantum walk from a single vertex, rather than an initial superposition corresponding to the quantum analogue of a stationary distribution. 
Moreover, often one can reduce the problem of finding marked elements to that of their detection, 
via a dichotomic search, as we will do. But this could also raise unexpected difficulties and interesting research questions. See for instance~\cite{AGJK20,AGJ21}.

In our work we adapt the above (detection) framework to the distributed setting. By doing so, we design for the first time a general framework for exploiting quantum walks on an unknown network, and we provide new quantum advantage for well-known and studied problems in distributed computing.

Let us also note here that the framework introduced in~\cite{mnrs} has already been exported to design a communication-efficient distributed leader election algorithm for {\em diameter-two networks}~\cite{DMP25}, but with a completely different implementation. The walk was not performed directly on the network structure, but on some virtual structure internal to some selected nodes. The network was only used to perform some other computations, such as preparing the network, or checking that a node is marked. Thus, the walk was never performed over the network itself.
However, recently in~\cite{legall2025} and independently of this work, the quantum solution of~\cite{LLL24} was shown to be implemented distributedly on the welded tree itself, providing an exponential quantum speed-up on the message complexity for the related routing problem, and paving the way for potential speed-ups of quantum walks on the network itself. In this solution, the walk is performed directly on the network itself.
However, the graph structure is known in advance, and only the port numbering remains unknown. Therefore, the generalization of such speed-ups is unclear. In particular, their approach relies on the specific solution of~\cite{LLL24}, leaving open the possibility to adapt any further sequential quantum algorithms based on quantum walks. Furthermore, they use a very rare possibility of~\cite{LLL24}, that is, the walk itself collapses with inverse polynomial probability to the target node. We will not have this situation, and thus we will need to deal with more technical issues, such as implementing distributedly quantum phase estimation, or in fact a weaker but sufficient procedure that we call quantum phase detection.

\smallskip

\noindent {\bf Quantum graph algorithms in the centralized setting.} The work of \cite{DHHM06} gave near-optimal bounds on  quantum {\em query complexity} for graph-theoretic problems such as connectivity,  strong connectivity,
MST, and Single Source Shortest Paths (SSSP). They show that the
query complexity of connectivity and MST is $\Theta(n^{1.5})$ in the adjacency matrix input model, but also that in the adjacency {\em array} model, the query complexities of the two problems differ: that of connectivity is $\Theta(n)$ whereas that of MST is $\Theta(\sqrt{mn})$.\footnote{One may wonder why our generic reduction lemma (cf. Lemma \ref{lem:CentralizedToDistributedLowerBound}) does not imply an $\tilde{\Omega}(\sqrt{mn})$ quantum message complexity lower bound for MST, which would in turn contradict our $\tilde{O}(n)$ message complexity upper bound for that very same problem. The reason is that the $\tilde{\Omega}(\sqrt{mn})$ query complexity lower bound in the adjacency array model arises due to the need to discover the weights of the edges by querying the adjacency array. In contrast, within the distributed setting (and this is a common assumption), nodes initially know the weights of all of their incident edges.} They achieve optimal query complexity upper bounds via Grover-type search algorithms, but these upper bounds do not translate to the distributed setting. For example, among other things, the connectivity upper bound would suffer from a significant slowdown in the distributed setting; for a similar slowdown, see our discussion regarding the locality cost in Section \ref{subsec:high-levelLE}. The corresponding lower bounds are obtained via the adversary method of Ambainis \cite{Ambainis02}, which we also use. However, their $\Omega(n)$ (strong) connectivity query complexity lower bound weakens to an $\Omega(D)$ query complexity lower bound for diameter-$D$ graphs, whereas we seek an $\Omega(n)$ query complexity lower bound for graphs of diameter 3 and more.

\section{Distributed Quantum Walks from Electric Networks} 
\label{sec:distributedQuantumWalk}

We first present the framework of electric networks~\cite{belovs2013,BelovsCJKM13} for designing efficient quantum algorithms using quantum walks, in the sequential setting. In more detail, these quantum walks can be used to efficiently detect the presence of certain marked nodes within some graph, but this is easily converted, at an extra $O(\log n)$ cost, into a quantum algorithm that finds such marked nodes.

Next, we show how to export this framework to the distributed setting, from encoding and executing one quantum walk in some distributed network, to ensuring multiple quantum walks, one per subgraph, can execute concurrently --- that is, without interference --- when the subgraphs are ``near-disjoint''. 

\subsection{Sequential Quantum Walks from Electric Networks}
\label{electric}
\paragraph{Electric Networks.} An electric network $(G,w)$ is a simple undirected graph  $G=(V,E)$ with positive edge-weights $(w_e)_{e\in E}$. We denote by $W$ some given upper bound on its total weight $\sum_{e\in E} w_e$. We also consider a distinguished set of marked vertices $M\subseteq V$ and a vertex $r\in V$ (that could be considered as the root of $G$). Depending on the setting, we also denote such a rooted network by $(G,w,r)$ or $(G,w,r,M)$ in the presence of marked vertices.

Consider any unit flow $(f_e)_{e\in E}$ on $G$ from $r$ to $M$. Formally, the flow is defined on the (oriented) edges of $G$ and satisfies $f_{v,u}=-f_{u,v}$. Moreover, one unit of flow is injected in $r$, that is $\sum_{v:\{r,v\} \in E} f_{r,v}=1$. All flow exits on $M$, that is $\sum_{u\in M,v:\{u,v\}\in E} f_{v,u}=1$. Lastly, all other nodes $u\not\in M\cup\{r\}$ preserve the flow, that is $\sum_{v:\{u,v\}\in E} f_{u,v}=0$. We denote by $R(f)=\sum_{e\in E} f_e^2/w_e$, the \emph{energy} of the flow.
The minimum possible value $\reff(r,M)$ of $R(f)$ is in fact \emph{effective resistance} between $r$ and $M$ on the electric network defined by $G$, where each edge $e$ has conductance $w_e$. When there is no ambiguity we simply write $\reff$.

One can show (see~\cite{belovs2013}) that it is possible to distinguish between $M=\emptyset$ and $M\neq\emptyset$ using $O(\sqrt{RW})$ steps of a quantum walk on $G$, where $R$ is an upper bound on $\reff$ when $M\neq\emptyset$. First, we provide in \Cref{qrw-thm} a specific version of that result. We then adapt it to the distributed setting in \Cref{dqrw-thm}.

\paragraph{Quantum Walk.} For now, we show how to define that quantum walk. Formally, it is the quantum analogue of a random walk on the (oriented) edges of $G$, one vertex indicating the current position, and the other one the next move to perform. 
For technical reasons, quantum walks are often defined on bipartite graphs, by eventually duplicating the original graph. We will not do this, except for the distinguished vertex $r$, for which we will simulate a connection to an artificial sibling node $\bar r$.
Similar variants of this approach have been used for instance in~\cite{apers22,hari25,zur25}.

We warn the reader that in the following definition of quantum walk, and later on in the analysis of the walk, most of the \emph{vectors} are not normalized when they are used to define vector spaces. On the other hand, \emph{quantum states} (of our algorithms) will always be normalized.
\begin{definition}[Quantum walk (adapted from~\cite{belovs2013})]\label{qwalk}
Let $(G,w,r) = (V_r,E_r)$ be some rooted electric network and let $C_1,R\geq 1$ be some constants. Let $\bar r$ be some artificial sibling node to $r$. Then, that electric network $(G,w,r)$ defines the (quantum) \emph{walk space} 
$$H=\Span(\ket{r, \bar r}, \ket{\bar r, r}, \ket{u,v},\ket{v,u}: \{u,v\} \in E_r).$$ 

Next, define $H_r=\Span(\ket{\bar r}, \ket{v}: \{r,v\}\in E_r)$ 
and $H_u=\Span(\ket{v}: \{u,v\} \in E_r)$ for $u\in V_r, u\neq r$.
Moreover, for all $u\in V_r$, define $D_u$ as the orthogonal reflection about $(\ket{\psi_u})^\perp$ in $H_u$,
where vectors $\ket{\psi_u}$ are defined by
$$
\ket{\psi_r}=\frac{1}{\sqrt{C_1R}}\ket{\bar r}+\sum_{v \in V_r:\{r,v\}\in E_r} \sqrt{w_{rv}}\ket{v},
\quad \text{and} \quad
\ket{\psi_u}=\sum_{v \in V_r:\{u,v\}\in E_r} \sqrt{w_{uv}}\ket{v}, \text{ for $u\in V_r, u\neq r$}.$$
Then the \emph{walk operator} $U(R)$ on $H$ is 
$$U(R)=(-\swap) \times D, 
\quad\text{where}\quad \swap\ket{u,v}=\ket{v,u} \text{ and }D=\sum_{u\in V_r} \ketbra{u} \otimes D_u.$$
Since $D$ acts as identity on $\ket{\bar r,r}$, we also define
$D_{\bar r}=\mathrm{Id}_{H_{\bar r}}$ where $H_{\bar r}=\Span(\ket{r})$. 
\end{definition}

Now we define the notion of quantum walk in presence of marked vertices.
\begin{definition}[Quantum walk with marked vertices]\label{qwalkmarked}
For a rooted electric network $(G,w,r,M)$ with marked vertices,
we simply adapt \Cref{qwalk} with
$D_u=\mathrm{Id}_{H_u}$ when $u\in M$,
leading to an operator $U(R,M)$. Observe that  $U(R)=U(R,\emptyset)$.
\end{definition}

One of the main properties of the quantum walk operator $U(R,M)$ relies on the overlap of its $1$-eigenspace with a well-chosen starting state $\ket{\sigma}$. We will take $\ket{\sigma}=\tfrac{1}{\sqrt{2}}(\ket{r,\bar r}-\ket{\bar r,r})$.
In~\cite{belovs2013}, it was observed that this overlap is large when $M\neq\emptyset$ and small when $M=\emptyset$.
(This is formally stated in \Cref{notone,one} of \Cref{analysis}.)

\paragraph{Quantum Phase Detection.}
Distinguishing between these two cases is usually done via Quantum Phase Estimation. However, this procedure requires a quantum control on the number of applications of the operator $U$ (i.e. the number of walk steps), which could be quite hard to implement in the distributed setting.
Hence, to anticipate our distributed implementation, we replace Quantum Phase Estimation with a simpler tool, originally used in~\cite{Szegedy04}.
Indeed, we simply need a procedure to detect whether an eigenvector has eigenvalue $e^{i\alpha}$ with $\alpha\neq 0$. Thus, we design a simpler variant of phase estimation, which uses not only less quantum bits but more importantly, needs only a classical control of the number of iterations of $U$. We call this procedure \emph{Quantum Phase Detection} (QPD), which we now define precisely but we postpone its analysis to \Cref{analysis} (see \Cref{qpd}):
\begin{quote}
\textbf{QPD$(U,T)$}
\begin{enumerate}
    \item Input: Quantum state $\ket{\phi}$ and unitary map $U$ that can be applied to $\ket\phi$
    \item Take $t$ uniformly at random in $\{1,2,\ldots, T\}$
    \item On a fresh qubit $\ket{0}$, called control qubit, apply $H$
    \item Controlled on this qubit, apply $t$ times $U$ on $\ket{\phi}$
    \item Apply $H$ on the control qubit
    \item Measure the controlled qubit and return the resulting bit
\end{enumerate}
\end{quote}

\paragraph{Detecting Marked Elements.} The procedure QPD can then be used to detect whether the starting state has a large or small overlap on the $1$-eigenspace of the walk, and thus to detect whether $M\neq \emptyset$. This lead us to the following statement from~\cite{belovs2013}, that we adapt to our setting and to QPD.
\begin{theorem}\label{qrw-thm}
Let $0<\delta<1$ and $R,W \geq 1$.
Let $(G,w,r,M)$ be a rooted electric network with marked vertices.
Assume that $\sum_{e\in E_r}w_e\leq W$ and $\reff(r,M) \leq R$ when $M\neq\emptyset$, and let $T\geq 80 \sqrt{(1/2)+9 R W}$, $C_1=9$. Then, QPD$(U(R,M),T)$ on input state $\ket{\sigma}=\tfrac{1}{\sqrt{2}}(\ket{r,\bar r}-\ket{\bar r,r})$:
\begin{enumerate}
    \item Outputs $1$ with probability at least $3/10$ when $M=\emptyset$;
    \item Outputs $1$ with probability at most $1/10$ when $M\neq\emptyset$;
    \item Uses at most $T$ steps of the quantum walk $U(P,R,M)$.
\end{enumerate}
\end{theorem}
\begin{proof}
We use in the proof \Cref{one,notone,qpd}, stated and proven in \Cref{analysis}. 

If $M\neq \emptyset$, by \Cref{one}, the overlap of $\ket{\sigma}$ with the $1$-eigenspace of $U(R,M)$ is at least $\sqrt{9/10}$. Moreover, by \Cref{qpd}, QPD outputs $1$ with probability $0$ for $1$-eigenvectors. Therefore, on input $\ket{\sigma}$, QPD outputs $1$ with probability at most $1-(\sqrt{9/10})^2=1/10$.

If $M=\emptyset$, by \Cref{notone},  the overlap of $\ket{\sigma}$ with the span of $e^{i\alpha}$-eigenvectors for $|\alpha|\leq \theta$, where we set $\theta=1/(C_2\sqrt{(1/2)+C_1RW})$ and $C_2 = 4$, is at most $1/4$.
So, the overlap with the span $e^{i\alpha}$-eigenvectors  for $|\alpha|> \theta$ is at least
$1-(1/4)^2> 9/10$. Since $T\geq 80 \sqrt{(1/2)+9 R W}=20/\theta$, by \Cref{qpd}, QPD outputs $1$ with probability at least $2/5$ for $e^{i\alpha}$-eigenvectors such that $|\alpha|\geq \theta$. 
Therefore, on input $\ket{\sigma}$, QPD outputs $1$ with probability at least $(9/10)(2/5)>3/10$.
\end{proof}

\subsection{Distributed Implementation of a Quantum Walk Step}
\label{subsec:distQuantumWalkImplementation}

We show, in this part, how to implement the sequential quantum walk described in Section \ref{electric} within (some part only of) a distributed network. More precisely, we first describe how nodes participating in some quantum walk know their local view of that walk's defining electric network. Next, we show how to perform one step of the corresponding quantum walk.

\paragraph{Implicit electric network}
From now on, the electric networks we consider are some subgraphs $G(r,\tk) = (V_r,E_r)$  
of the communication graph $G = (V,E)$ augmented with an edge-weight function over $E_r$, a designated root vertex $r \in V_r$ and a subset of marked vertices $M(\tk) \subseteq V_r$.  

Such an electric network can be distributedly known by nodes of $V$ as follows. The special node $r\in V_r$ has the knowledge of being distinguished, and $r$ owns a special token value $\tk$ of size $\mathrm{polylog}(|V|)$ bits. (Such tokens take values in some subset $\mathcal{T}$ among all possible messages.) In practice, this token is simply a unique identifier (randomly) generated by $r$.
In addition, each node $v \in V$ is given some local information such that, together with the token value $\tk$, it can compute its own status as marked (i.e., whether $v \in M(\tk)$) or unmarked. (In practice, the local information contains some unique identifier, and $v$ is marked if say, that identifier is different from the token's identifier.) Additionally, when $v$ is unmarked, its local information together with $\tk$ allows $v$ to compute the local structure of $G(r,\tk)$, that is, which incident edges are in $E_r$ as well as their weights. (On the other hand, a marked node need not know its local structure, as it will only reply to received messages.)  
The resulting electric network is denoted $G(r,\tk)$, and called the \emph{implicit electric network rooted in $r$ with token $\tk$}.

\begin{definition}[Implicit electric network]
An \emph{implicit electric network} on the communication graph $G$ with token in $\mathcal{T}$ is a collection of 
Boolean functions $(M_u)_{u\in V}$, defined on input $\tk \in \mathcal{T}$, and indicating which nodes are marked, along with a collection of non-negative weight functions $(w_u)_{u\in V}$, defined on inputs $(\tk,p) \in \mathcal{T} \times [\deg(u)]$.  
Furthermore, these weight functions must satisfy $w_u(\tk,p)=w_v(\tk,q)$, for every token $\tk\in \mathcal{T}$ and emission ports $p=u\to v, q=v\to u$.

Given any $r \in V, \tk \in \mathcal{T}$, we denote by $G(r,\tk) = (V_r,E_r)$ the underlying electric circuit network 
restricted to edges $\{u,v\}$ with weight  $w_{uv}=w_u(\tk,u\to v)>0$, and furthermore to the connected component of $r$, with marked vertices $M(\tk)=\{v\in V_r:M_v(\tk)=\mathrm{true}\}$.
\end{definition}
\textbf{Remark.} Let $u\in M(\tk)$. In that case $D_u=\mathrm{Id}_{H_u}$. Therefore the weight functions $w_u$, for $u\in M(\tk)$, are useless in order to simulate the quantum walk $U(R,M(\tk))$ on $G(r,\tk)$. Thus, we can relax the above definition so that the weights are only given to the non marked vertices.

\paragraph{Encoding walk states in distributed networks}
Consider some quantum walk corresponding to some rooted implicit electric network $G(r,\tk) = (V_r,E_r)$. By definition, the states of that sequential walk are $H=\Span(\ket{r, \bar r}, \ket{\bar r, r}, \ket{u,v},\ket{v,u}: \{u,v\} \in E_r)$. They represent the walk being on some (oriented) edge of $E_r \cup \{r,\bar r\}$.

Within a distributed network, as independently proposed in~\cite{legall2025}, we can implement states $\ket{u,v}$ and $\ket{v,u}$ for any edge $\{u,v\} \in E_r$ by having, respectively, the token $\tk$ in the port $u\to v$ of $u$ or in the port $v \to u$ of $v$. As for states $\ket{r,\bar r}$ or $\ket{\bar r,r}$, these can be implemented via virtual port registers local to node $r$ (since $\{r,\bar r\}$ is, after all, a virtual edge not in the communication graph).

More formally, let node $r$ contain four extra local registers, named $r \to \bar r$, $r \gets \bar r$, $\bar r \to r$, and $\bar r \gets r$, and that are all initialized to $\ket{\vac}$. These correspond to (virtual) emission and reception port registers for communicating over $\{r,\bar r\}$, which $r$ will use to locally simulate $\bar r$. Given these registers, we can encode $\ket{u,v}$, for any edge $\{u,v\} \in E_r \cup \{r,\bar r\}$, in the distributed network as: 
$$\ket{u,v}'= \ket{\tk}_{u\to v}\otimes\ket{\vac}_{u\gets v} \otimes \ket{\vac}_{v\to u}\otimes\ket{\vac}_{v\gets u}
\bigotimes_{\{x,y\}\in E_r \cup \{r,\bar r\} \setminus \{u,v\}} \ket{\vac}_{x\to y}\otimes\ket{\vac}_{x\gets y} \otimes \ket{\vac}_{y\to x}\otimes\ket{\vac}_{y\gets x},$$
where $\tk$ is the special token value, initially computed by $r$. Naturally, this induces the \emph{distributed walk space} $H' = \Span(\ket{r,\bar r}', \ket{\bar r,r}', \ket{u,v}',\ket{v,u}' : \{u,v\} \in E_r)$.

Note that $H'$ consists of the superposition of configurations where only one node is going to send a message, and that message is $\ket{\tk}$. In particular, $\ket{r,\bar r}'$ and $\ket{\bar r,r}'$ are configurations where that token lies in the corresponding virtual emission port registers (respectively, $\ket{\tk}_{r\to \bar r}$ and $\ket{\tk}_{\bar r \to r}$) local to $r$. The associated communication is fictive and locally simulated by $r$. In fact, node $r$ locally simulates $\bar r$, but as can be seen from \Cref{qwalkmarked}, this amounts to the identity outside of the communication operations.

\paragraph{Realizing one step of the walk}
We can now provide a simulation of the quantum walk operator $U(R,M)$ in space $H'$, as stated by the lemma below.
\begin{lemma}\label{dqw}  
Let $G=(V,E)$ be a distributed network with local inputs.
Let $G(r,\tk)$ be an implicit electric network rooted in some $r\in V$, with token $\tk$ and subset of marked elements $M$, and let $U(R,M)$ be the corresponding quantum walk operator.
Then, one step of $U(R,M)$ can be simulated on $G$ with one message and one round of communication.
\end{lemma}
\begin{proof}
We want to implement $U(R,M)= (-\swap)\times D$ on the distributed network.
The operator $D$ is going to be communication-free, whereas $\swap$ will consume one message. 

We start by $\swap$, since it is the simplest operation. Implementing that operation distributedly, on the space of (distributed) walk states, boils down to executing: $\ket{u,v}' \stackrel{\swap'}{\mapsto} \ket{v,u}' $.
This can be implemented easily in the distributed network, via
\begin{enumerate}
    \item A communication step (realized by operator $\send$);
    \item A virtual communication step (realized by some operator $\virtualsend$), where $r$ locally simulates $\send$ on its virtual port registers;
    \item A local computation step (realized by some operator $\reverse$), where we reverse the direction of each communication port. 
\end{enumerate}

More precisely, recall first that the $\send$ operator (executed by the distributed network every round) ensures that, for any emission register $u \to v$, its contents are exchanged with that of the reception register $v \gets u$. As for the operator $\virtualsend$, it exchanges the contents of the virtual emission register $r \to \bar r$ with the virtual reception register $\bar r \gets r$, and that of the virtual emission register $\bar r \to r$ with the virtual reception register $r \gets \bar r$. Finally, the operator $\reverse$ is defined so that for any reception register $u \gets v$, its contents are exchanged with the emission register $u \to v$. Note that we can locally implement $\reverse$ since any node $u$ knows which of its emission and reception registers correspond to the same node (i.e., here to $v$), even if $u$ cannot identify that node. Indeed, $u$ even initially assigns the same port number (in 1 to $\deg(u)$) to these two registers. 

Thus, the above 3-step process for $\swap$, where we only explicitly write out the port registers $u\to v$, $v\gets u$ and $v\to u$ (the other ones remain in state $\ket{\perp}$), leads to:
\begin{align*}
    \ket{u,v}'=\ket{\tk}_{u\to v}\ket{\vac}_{v\gets u}\ket{\vac}_{v\to u} &\stackrel{\scriptsize \eqmakebox[a][c]{$\virtualsend \times \send$}}{\mapsto} \ket{\vac}_{u\to v}\ket{\tk}_{v\gets u}\ket{\vac}_{v\to u} \\
    &\stackrel{\scriptsize \eqmakebox[a][c]{$\reverse$}}{\mapsto} \ket{\vac}_{u\to v}\ket{\vac}_{v\gets u}\ket{\tk}_{v\to u}=\ket{v,u}'
\end{align*}

Next, we implement $-\swap$ in our distributed setting, or in other words, we execute $\ket{u,v}' \stackrel{(-\swap)'}{\mapsto} -\ket{v,u}' $. To do so, it suffices to apply, after the above 3-step process, the operator $\flip$. That operator realizes, on each (real and virtual) emission register $u\to v$, the following: for any token message $\tk \neq \vac$, $\flip\ket{\tk}_{u \to v}=-\ket{\tk}_{u \to v}$, and otherwise (i.e., on orthogonal states including $\ket{\vac}$), $\flip$ is the identity. 
In summary, we get $$(-\swap)'=\flip\times \reverse\times\send.$$

Finally, we focus on operator $D'=\left(\sum_{u\in V} \ketbra{u}\otimes D_u\right)'$. Recall that the $\ketbra{u}\otimes D_u$ operators work on orthogonal subspaces of $H$, and thus that their distributed implementations, $D'_u=(\ketbra{u}\otimes D_u)'$, act on the orthogonal direct sum of subspaces $(\ket{u}\otimes H_u)'$ of $H'$.
Hence, to implement $D'$, intuitively each node $u$ only needs to simulate $D'_u$ when it holds the token, and the identity otherwise. It is not too difficult to check that this is a unitary map, even though $D'_u$ depends also on the received token. For completeness, we decompose this operator $D'_u$ into simple unitary steps.

Consider any node $u \in V$, and we shall need the following to describe $D'_u$. In addition to the emission and reception registers, assume $u$ holds two additional work registers, named $X$ and $Y$. Moreover, let $\deg'(u)$ be the degree of $u$ whenever $u \neq r$, and $\deg(u)+1$ otherwise. Then, for any node $u$, the space of the $X$-register is $\Span(\ket{p} : 1 \leq p \leq \deg'(u))$, where emission ports have been identified with their port numbers (and port $\deg'(r)$ corresponds to $r$'s virtual port to $\bar r$). Let $H''_u(r,\tk)$ be its subspace restricted to port numbers corresponding to the edges incident to $u$ that are in $G(r,\tk)$ (and also to port $\deg'(r)$ if $u = r$), and let $D_u''(r,\tk)$ be the operator similar to $D_u$, but that applies to $H_u''(r,\tk)$. The space of the second register is the vector space generated by all possible token messages. Finally, for any port number $1 \leq p \leq \deg'(u)$, we define a memory-register exchange operator $L_p$ such that $L_p:\ket{0}_X\ket{\vac}_Y\ket{\tk}_{u\to v}\mapsto \ket{p}_X\ket{\tk}_{Y}\ket{\vac}_{u\to v}$, \ \  $L_p:\ket{p}_X\ket{\tk}_{Y}\ket{\vac}_{u\to v} \mapsto\ket{0}_X\ket{\vac}_Y\ket{\tk}_{u\to v}$,\ \  
and $L_p$ extends arbitrarily (for instance as identity) on the orthogonal space (which includes, in particular, when an emission register holds a non-token message), in such a way that $L_p$ remains unitary.
Then, $u$ implements $D'_u$ as follows:
\begin{enumerate}
    \item Assume registers $X$ and $Y$ are at initial value $\ket{0}_X$ and $\ket{\vac}_Y$ .
    \item \label{step2a} Sequentially, for each (real or virtual) emission register $u\to v$ with port numbering $1\leq p \leq \deg'(u)$, apply $L_p$ (to the work, emission and reception registers held by $u$).
    \item\label{step2b} If the register $Y$ holds some $\tk\neq\vac$: Apply $D_u''(r,\tk)$ on space $H_u''(r,\tk)$ of register $X$.\\
     \emph{Note that $D_u''(r,\tk)$ and $H_u''(r,\tk)$ can be computed locally using $\tk$.}
    \item Apply again Step~\ref{step2a}.\\
    \emph{Registers $X,Y$ are back to their initial values $\ket{0}_X$ and $\ket{\vac}_Y$.}
\end{enumerate}
Observe that Step~\ref{step2b} is performed just for one value of the $X$-register, or in other words, for one port, but in superposition. Indeed, there is a single token traveling through the network (in superposition). Moreover, note that on walk state $\ket{\bar r, r}'$, node $r$ performs the identity operator for the diffusion.
\end{proof}

\subsection{Distributed Detection of Marked Elements}

With our distributed implicit implementation of quantum walks, we can now give a distributed adaptation of \Cref{qrw-thm}.
\begin{theorem}\label{dqrw-thm}  
Let $0<\delta<1$ and $R,W \geq 1$.
Let $G=(V,E)$ be a distributed network with local inputs.
Let $G(r,\tk)$ be an implicit electric network rooted in some $r\in V$ with token $\tk$ and let $M$ be some set of marked elements such that
(1) $G(r,\tk)$ has total weight at most $W$;
(2) either $M=\emptyset$ or $\reff(r,M) \leq R$.

There is a quantum distributed algorithm such that $r$ can decide whether $M=\emptyset$ 
with bounded error $\delta$ and message complexity $O(\log (1/\delta) \times \sqrt{RW})$.
\end{theorem}
\begin{proof}
The proof relies on \Cref{qrw-thm}. The only difference (with \Cref{qrw-thm}) lies in the implementation of QPD, which must be executed on the distributed network, and within which, the simulation on $G$ of $U(R,M)$ is done distributedly as explained in~\Cref{dqw}.

More precisely, root $r$ is in charge of simulating the centralized algorithm $QPD(U(R,M),T)$.
Within which, when it comes down to applying $U$ (and only for that), the network is used to simulate (a step of) the walk as described in \Cref{dqw}, while the control qubit remains in $r$.
Each step of the quantum walk is executed synchronously by the full network (on $G(r,\tk)$) for $t$ steps.

In order to stop the propagation of the walk, we formally include the value of $t$ into the token, 
so that the network can stop the walk after exactly $t$ steps. The remarkable property of QPD$(U(R,M),T)$ is that the root $r$ does not need to take back the control on the walk state. Indeed, the last steps of QPD (Steps~5 and~6) only use the control qubit, which remains in $r$.
\end{proof}

We now consider the possibility of several quantum walks being implemented in parallel, and possibly concurrently over the same parts of the network.
Let us first consider the case that the different concurrent walks are edge-disjoint, but not necessarily vertex-disjoint.
\begin{corollary}\label{cdqrw-cor}  
Let $G=(V,E)$ be a distributed network with local inputs, with some subset $A\subseteq V$ of \emph{distinguished vertices} and a set of possible tokens $\mathcal{T}$.
Let $0<\delta<1$.
For each $r\in A$,
let $G_r=G(r,\tk_r)$ be some implicit electric network rooted in $r$ with token $\tk_r\in\mathcal{T}$, and $M_r$ be a set of marked elements such that
(1) $G_r$ has total weight at most $W_r$;
(2) either $M_r=\emptyset$ or $\reff(r,M_r) \leq R_r$;
(3) $(G_r)_{r\in A}$ are pairwise edge-disjoint;
(4) tokens $\tk_r$ are pairwise distinct.

There is a quantum distributed algorithm such that each $r\in A$  decides whether $M_r=\emptyset$ with bounded error $\delta$. Moreover the total message complexity is $O(\log (1/\delta) \times \sum_{r\in A} \sqrt{R_rW_r})$, and the total
 round complexity is $O(\log (1/\delta) \times \max_{r\in A} \sqrt{R_rW_r})$.
\end{corollary}
\begin{proof}
The proof consists in adapting the proof of \Cref{dqrw-thm}, so that applying the quantum walk operator $U$ once now executes, for each root $r \in A$ and token $\tk_r \in \mathcal{T}$, one step of the corresponding walk, and such that $U$ can be (efficiently) simulated in a distributed fashion. (Note that the different tokens may execute walks of various lengths. Dealing with this is simple: any node holding a token finished with its walk simply does not execute the below operators on that token.)
The simulation of $U$ does not directly follow from that given within the proof of \Cref{dqw}. The main difference --- and one that could lead to multiple walks interfering with each other --- is that now two (or more) tokens could be held by the same vertex, say $u \in V$. Hence, we need to modify how $u$ implements the walks' diffusion operators. However, intuitively, these operators do not interfere with each other since different tokens can never use the same emission ports of $u$. 

To show this intuitive statement, we need the following preliminaries. 
First, observe that each quantum walk's electric network is completely determined by its token's content, $\tk$, as we can always put the (randomly chosen) unique identifier of the network's root $r$ within $\tk$, and it remains the case (with high probability) that these extended tokens are pairwise distinct for any two roots in $A$. Hence, we denote that electric network by $G(\tk) = (V_\tk,E_ \tk)$, in which case the space of the sequential walk on $G(\tk)$ is $H_\tk= \Span(\ket{r,\bar r}, \ket{\bar r,r}, \ket{u,v}, \ket{v,u}: \{u,v\} \in E_\tk)$ (where $r$ is given by the value of $\tk$). 
Moreover, Definition \ref{qwalk} applied to $G(\tk)$ gives an operator $U_\tk =  (-\swap_\tk) \times D_\tk$ for the walk, where $\swap_\tk$ represents the $\swap$ operator limited to edges in $G(\tk)$, and $D_\tk$ the diffusion operator on $G(\tk)$. 
Now, if we observe the quantum walks associated to roots in $A$, then since the implicit electric networks $(G_{\tk_r})_{r \in A} = (G_r)_{r \in A}$ are pairwise edge-disjoint, the $(H_{\tk_r})_{r \in A}$ are disjoint vector spaces and the multiple concurrent quantum walks execute on $\bigotimes_{r \in A} H_{\tk_r}$. Furthermore, it holds that $U = \bigotimes_{r \in A} U_{\tk_r} =  \bigotimes_{r \in A} \left( (-\swap_{\tk_r}) \times D_{\tk_r}\right) = \bigotimes_{r \in A} (-\swap_{\tk_r}) \times \bigotimes_{r \in A} D_{\tk_r}$.

Next, note that with multiple quantum walks in the distributed network, the corresponding distributed encoding --- which must now capture the positions of all tokens --- necessarily changes compared to the case of a single quantum walk. Hence, we redefine the distributed encoding, albeit with some abuse of notations. More precisely, for any token message $\tk$ and for any edge $\{u,v\} \in E_\tk$, we define the distributed walk state $\ket{u,v}'$ as:
$$\ket{u,v}'= \ket{\tk}_{u\to v}\otimes\ket{\vac}_{u\gets v} \otimes \ket{\vac}_{v\to u}\otimes\ket{\vac}_{v\gets u}
\bigotimes_{\{x,y\}\in E_\tk \cup \{r,\bar r\} \setminus \{u,v\}} \ket{\vac}_{x\to y}\otimes\ket{\vac}_{x\gets y} \otimes \ket{\vac}_{y\to x}\otimes\ket{\vac}_{y\gets x},$$

where $r$ is determined by the value of $\tk$, and has virtual port registers associated to the virtual edge $\{r, \bar r\}$.
The distributed walk state captures the position, within the distributed network, of the token corresponding to $G(\tk)$ only. As such, we can also define $H_{\tk}' = \Span(\ket{r,\bar r}',\ket{\bar r,r}',\ket{u,v}', \ket{v,u}' : \{u,v\} \in E_\tk)$, but also some operators $U_\tk'$, $(-\swap_\tk)'$ and $D_\tk'$ on that space, in the natural sense, from $U_\tk$, $\swap_\tk$ and $D_\tk$.
From these individual distributed walk spaces, we obtain the ``global'' walk space via a tensor product: $\bigotimes_{r \in A} H_{\tk_r}'$. In other words, states in $\bigotimes_{r \in A} H_{\tk_r}'$ capture the positions of the $|A|$ distributed quantum walks within the distributed network. 

How do we simulate $U$ in a distributed network, or in other words, how do we implement $\bigotimes_{r \in A} U_{\tk_r}'$? First, observe that if we take $\virtualsend$, $\reverse$ and $\flip$ as defined in the proof of \Cref{dqw} (and where $\virtualsend$ applies to the virtual port registers of any distinguished vertex $r \in A$), then the operator $\flip \times \reverse \times \virtualsend \times \send$ is, in fact, the distributed implementation of $\bigotimes_{r \in A} (-\swap_{\tk_r})$, that is, $\flip \times \reverse \times \virtualsend \times\send=\bigotimes_{r \in A} (-\swap_{\tk_r})'$. Hence, it remains to implement $\bigotimes_{r \in A} D_{\tk_r}'$.

To do that, note first that for any token message $\tk$, $D_\tk = \sum_{u \in V_\tk} \ketbra{u}{u} D_{\tk,u}$, where Definition \ref{qwalk} applied on $G(\tk)$ gives for any $u \in V_\tk$, subspace $H_{\tk,u}$, state $\ket{\psi_{\tk,u}}$ and operator $D_{\tk,u}$ --- these respectively capture the neighbors of $u$ within $G(\tk)$, a (edge) weighted quantum superposition over these neighbors, and an orthogonal reflection about $(\ket{\psi_{\tk,u}})^\perp$ in $H_{\tk,u}$. Then, for any token message $\tk$, the operators $(\ketbra{u}{u} \otimes D_{\tk,u})_{u \in V_\tk}$ work on orthogonal subspaces of $H_\tk$, hence if we let the distributed implementation be $D_{\tk,u}' = (\ketbra{u}{u} \otimes D_{\tk,u})'$ then the operators $(D_{\tk,u}')_{u \in V_\tk}$ acts on the orthogonal direct sum of subspaces $(\ket{u}\otimes H_{\tk,u})'$ of $H_\tk'$. Hence, to implement $D_\tk'$, intuitively each node $u$ need only simulate $D_{\tk,u}'$ when $u$ holds token $\tk$, and the identity otherwise.

To extend such a simulation to that of $\bigotimes_{r \in A} D_{\tk_r}'$, then each node $u \in V$ need only only simulate $\bigotimes_{r \in A} C_{\tk_r,u}$, where $C_{\tk,u} = D_{\tk,u}'$ when $u$ holds token $\tk$, and the identity otherwise. We next show how this can be done, for any node $u \in V$. In short, we adapt the description given in the proof of \Cref{dqw} for the similar operator, but with one major change: an outer loop iterating over all possible token messages. This, along with the fact that token are defined on pairwise edge-disjoint electric networks, ensures that when a node holds multiple walks' tokens, the corresponding quantum walks do not interfere with each other.  

More precisely, in addition to the emission and reception registers, assume $u$ holds two additional work registers, named $X$ and $Y$. Moreover, let $\deg'(u)$ be the degree of $u$ whenever $u \notin A$, and $\deg(u)+1$ otherwise.
The space of the $X$-register is $\Span(\ket{p} : 1 \leq p \leq \deg'(u))$, where emission ports have been identified with their port numbers. For any token message $\tk$, let $H_{\tk,u}''$ be the $X$-register's subspace restricted to port numbers corresponding to the edges incident to $u$ that are in $G(\tk)$ (and also to port $\deg'(u)$ if $u \in A$), and let $D_{\tk,u}''$ be the operator similar to $D_{\tk,u}$, but that applies to $H_{\tk,u}''$. The space of the second register is the vector space generated by all possible token messages. Finally, for any token $\tk \in \mathcal{T}$ and any port number $1 \leq p \leq \deg'(u)$, we define a memory-register exchange operator $L_{p,\tk}$ such that $L_{p,\tk}:\ket{0}_X\ket{\vac}_Y\ket{\tk}_{u\to v}\mapsto \ket{p}_X\ket{\tk}_{Y}\ket{\vac}_{u\to v}$,   \ \  $L_{p,\tk}:\ket{p}_X\ket{\tk}_{Y}\ket{\vac}_{u\to v }\mapsto\ket{0}_X\ket{\vac}_Y\ket{\tk}_{u\to v}$,\ \  
and $L_{p,\tk}$ extends arbitrarily (for instance as identity) on the orthogonal space (which includes, in particular, when an emission register holds any token message other than $\tk$, or also a non-token message), in such a way that $L_{p,\tk}$ remains unitary.
Then, $u$ implements $\bigotimes_{r \in A} C_{\tk_r,u}$ as follows.

\begin{enumerate}
    \item Assume registers $X$ and $Y$ are at initial value $\ket{0}_X$ and $\ket{\vac}_Y$ .
    \item Sequentially, for each possible value of $\tk \in \mathcal{T} \setminus \{\vac\}$, do:
    \begin{enumerate}
    \item \label{step2aa} Sequentially, for each (real or virtual) emission register $u\to v$ numbered by $1\leq p \leq \deg'(u)$ such that $p$ encodes an edge for $\tk$ --- in other words, an edge of $G(r,\tk)$ where $r$ follows from the value of $\tk$, or the virtual edge $\{r, \bar r\}$ --- apply $L_{p,\tk}$ (to the work, emission and reception registers held by $u$).
    \item\label{step2bb} If the register $Y$ holds  $\tk$: Apply $D_{\tk,u}''$ on space $H_{\tk,u}''$ of register $X$.\\
     \emph{Note that $D_{\tk,u}''$ and $H_{\tk,u}''$ are computed locally using $\tk$.} 
    \item\label{step2cc} Apply again Step~\ref{step2aa}.\\
    \emph{Registers $X,Y$ are back to their initial values $\ket{0}_X$ and $\ket{\vac}_Y$.}
\end{enumerate}
\end{enumerate}
\end{proof}

Now we would like to go a bit further by considering concurrent walks that could share some edges. In a classical (i.e. non-quantum) setting, one can handle the resulting congestion by augmenting either the number of rounds or the bandwidth (or the number of communicating ports per edge). But, in the quantum setting, we would face a critical synchronization issue: making messages with different arrival order interfere with each other.
Indeed, with a phase structure, the messages corresponding to some $i$th step of the quantum walk may arrive (from different neighbors) in different rounds of the same phase $i$, in which case they will not interfere with each other in the way described by the quantum walk framework.
However, putting these messages back in the same order is not a reversible operation, since it would ``erase" the initial order.

In the following corollary, we circumvent this issue for the specific case of concurrent edges leading to marked vertices, with a carefully designed synchronization process.
\begin{corollary}\label{cdqrw-cor2}  
Let $G=(V,E)$ be a distributed network with local inputs, 
with some subset $A\subseteq V$ of \emph{distinguished vertices}, and let $0<\delta<1$. 
For each $r\in A$,
let $G_r=G(r,\tk_r)$ be some implicit electric network rooted in $r$ 
with token $\tk_r$, and $M_r$ be a set of marked elements such that
(1) $G_r$ has total weight at most $W_r$;
(2) either $M_r=\emptyset$ or $\reff(r,M_r) \leq R_r$;
(3) $G_r,G_{r'}$ are edge-disjoint for each $r'\neq r$, except potentially on edges $\{u,v\}$ such that 
$\{u,v\}\cap(M_r\cup M_{r'})\neq\emptyset$;
(4) tokens $\tk_r$ are pairwise distinct.

There is a quantum distributed algorithm such that each $r\in A$  decides whether $M_r=\emptyset$ with bounded error $\delta$. Moreover the total message complexity is $O(\log (1/\delta) \times \sum_{r\in A} \sqrt{R_rW_r})$, and the total
 round complexity is $O(\log (1/\delta) \times \max_{r\in A} \sqrt{R_rW_r})$.
\end{corollary}
\begin{proof}
We adapt the proof of \Cref{cdqrw-cor} by incorporating a separate procedure to deal with marked vertices. To do so, we have all nodes virtually duplicate the ports toward marked nodes, and via these, locally simulate the quantum walk stepping over to some marked node (and back, by definition). As a result, we ensure that the distributed walk space remains the tensor product of disjoint vector spaces.
Via this procedure we blow up the number of rounds required to simulate one quantum walk step by a small constant (i.e., by 3).

Indeed, observe that for any edge $\{u,v\}$ shared between two graphs $G_r$ and $G_{r'}$, the condition $\{u,v\}\cap(M_r\cup M_{r'})\neq\emptyset$ means that at least one of the endpoints is connected to a marked element in either graph $G_r$ or $G_{r'}$. Assume w.l.o.g that the edge $\{u,v\}$ of $G_r$ leads to a marked node $v$ from $u$. Intuitively, such a marked node $v$ will be simulated locally in $u$ for the walk on $G_r$. 
However, how does $u$ know whether any port $p$ leads to a marked node $v$? Node $u$ can use two rounds to communicate with the other endpoint of a port $p$, and learn whether the other endpoint is marked. Formally, node $u$ ``queries'' $v$ by sending the value $\ket{\tk_r,0}$ to $v$, and $v$ responds with $\ket{\tk_r,b_v}$ where $b_v = 1$ if and only if $v$ is marked. (Although this querying sends the token content over edge $\{u,v\}$, the walk ``remains'' on $u$.) Now, while $u$ sends a single message over $u\to v$, $v$ may get several queries, but can answer to each of them, without interference, on the corresponding port in the next round.
After which, the last (third) round can be used to send the token over if the other endpoint is not marked, or to simulate that communication locally otherwise.

We now sketch the modified simulation. Each quantum walk step is no longer implemented in the distributed network by a single round, but by a phase of 3 rounds now. 

First, consider the case that the token $\tk$ is in some emission register of some non-marked node $u$ (this includes virtual port registers leaving from $u$). Then, $u$ first locally executes Steps~\ref{step2aa} and~\ref{step2bb}. Then, prior to executing Step~\ref{step2cc} and sending $\tk$ over some port $p=u \to v$, $u$ uses the first two rounds of the phase to communicate with node $v$ and learn whether $v$ is marked, as explained above. If $v$ is marked, then $v$ is simulated locally by $u$, and a virtual port is used (for Step~\ref{step2cc}) instead of $p$. The last and third round of the phase implements the (actual or virtual) token communication. (Observe that both endpoints of an edge could do the same procedure. But since the edges are full duplex, the messages will not be in conflict: they will use the different directions $u\to v$ and $v \to u$, and thus different channels.) 

Next, consider the case that the token $\tk$ is in some (virtual) emission register $u \to v$, for some marked node $u$. Note that $u$ must be simulated by $v$, and thus that virtual emission register is local to $v$. Moreover, $D_u=\mathrm{Id}$ since $u$ is marked. Therefore, the simulation consists in taking back (virtually) the token to $v$. Step~\ref{step2aa} is just performed in $v$ for a virtual node $u$ and a virtual port leading back to $v$. Then, Step~\ref{step2bb} does not involve any new vertex since $D_u=\mathrm{Id}$. So, Step~\ref{step2cc} is also performed locally and the token goes back to $v$.
\end{proof}

\subsection{Deferred Analysis for Sequential Quantum Walks}\label{analysis}
Below, the Effective Spectral Gap Lemma~\cite{LeeMRSS11} will allow us to study the spectral properties of the quantum walks we define in \Cref{electric}.
\begin{lemma}[Effective Spectral Gap Lemma~\cite{LeeMRSS11}]\label{effective-lemma}
Let $\Pi_A$ and $\Pi_B$ be two orthogonal projectors in the same vector space, and $R_A = 2\Pi_A-\mathrm{Id}$ and $R_B = 2\Pi_B-\mathrm{Id}$ be the reflections about their images.
Assume $P_\theta$, where $\theta \geq 0$, is the orthogonal projector on the span of the eigenvectors of $R_B\times R_A$ with eigenvalues $e^{i\alpha}$ such that $|\alpha| \leq \theta$. 
Then, for any vector $w$ in the kernel of $\Pi_A$, we have $\norm {P_\theta\Pi_B w} \leq \theta  \norm{w}/2$.
\end{lemma}

We can now state the two lemmas from~\cite{belovs2013} that will allow to use quantum walks $U(R,M)$ in order to decide whether $M\neq\emptyset$. We provide their proof for completeness, since we adapt them slightly for our setting.
\begin{lemma}[\cite{belovs2013}]\label{notone}
Assume that $M= \emptyset$. Set $\theta=1/(C_2\sqrt{(1/2)+C_1RW})$.
Then state $\ket{\sigma}=\frac{1}{\sqrt{2}}(\ket{r,\bar r}-\ket{\bar r,r})$ has a component of size at most $1/C_2$ on the span of eigenvectors of $U(R)$ with phase at most $\theta$.
\end{lemma}
\begin{proof}
Let $P_{\theta}$ be the orthogonal projector on the span of eigenvectors of $U(R)$ with phase at most $\theta/C_2$. 
We will bound $\norm{P_\theta \ket{\sigma}}$
using \Cref{effective-lemma}, with $R_A=D$ and $R_B=(-\swap)$.
Let $\Pi_A$ and $\Pi_B$ be the corresponding projectors 
such that $R_A=2\Pi_A-\mathrm{Id}$ and $R_B=2\Pi_B-\mathrm{Id}$.

Let  $\ket{w}$ be the vector
$$\ket{w}=\sqrt{C_1R}\left(\sum_{u\in V_r}\ket{u}\otimes\ket{\psi_u}\right).$$

Then, first observe that $D\ket{w}=-\ket{w}$ (since there are no marked elements) and thus $\Pi_A\ket{w}=0$. 
Next, we can expand $\ket{w}$ to
$$\ket{w}=\ket{r, \overline r}
+\sqrt{C_1R}\sum_{\{u,v\}\in E_r}(\sqrt{w_{uv}}\ket{u,v} + \sqrt{w_{vu}}\ket{v,u}).$$

In which case, since the second term of $\ket{w}$ is symmetric (as $w_{uv} = w_{vu}$), that second term is in the kernel of $\Pi_B$.
Therefore
$$\Pi_B \ket{w}=\Pi_B \ket{r,\bar r}=\frac{1}{2}(\ket{r,\bar r}-\ket{\bar r,r})=\frac{1}{\sqrt{2}}\ket{\sigma}.$$
Thus, since $\norm{\ket{w}}^2=1+2C_1RW$, we apply \Cref{effective-lemma} and conclude that
$$\norm{P_\theta \ket{\sigma}}=\sqrt{2}\norm{P_\theta \Pi_B\ket{w}}
\leq \sqrt{2}\theta \norm{w}/2= \theta\sqrt{(1/2)+C_1RW}= 1/C_2.$$
\end{proof}

\begin{lemma}[\cite{belovs2013}]\label{one}
Assume that $M\neq \emptyset$, $r\not\in M$,
and $R\geq \reff $.
Then the state $\ket{\sigma}=\frac{1}{\sqrt{2}}(\ket{r,\bar r}-\ket{\bar r,r})$
has a component of size at least $\sqrt{\frac{C_1}{1+C_1}}$ on the $1$-eigenspace of $U(R,M)$.
\end{lemma}
\begin{proof}
Let $f$ be some flow from $r$ to $M$ with energy $R(f)\leq R$.
Define the vector $$\ket{\phi}=\sqrt{C_1R}(\ket{r,\bar r}-\ket{\bar r,r})-\sum_{(u,v):\{u,v\}\in E_r}\frac{f_{u,v}}{\sqrt{w_{uv}}}\ket{u,v}.$$ 
Then the component of $\ket{\sigma}$ on the $\Span(\ket{\phi})$ satisfies
$$\frac{1}{\norm{\ket{\phi}}}\braket{\phi}{\sigma}
=\sqrt{\frac{2C_1 R}{2C_1R+2R(f)}}
\geq \sqrt{\frac{C_1}{1+C_1}}
.$$

We now show that $\ket{\phi}$ is a $1$-eigenvector of $U(R,M)$. 
In order to compute $D\ket{\phi}$, let us first compute $\braket{\psi_u,u}{\phi}$ for $u\not\in M$:
$$\braket{\psi_u, u}{\phi}=
\begin{cases} -\sum_{v:\{u,v\}\in E_r} f_{u,v} = 0,& \text{ when $u\neq r$};\\
1-\sum_{v :\{r,v\}\in E_r} f_{r,v} = 0,& \text{ when $u= r$}.
\end{cases}$$
Where the above equalities follow from the properties of the flow $f$.
Thus, $\ketbra{u}\otimes D_u$ acts as the identity on $\ket{\phi}$ when $u\not\in M$.
But by definition, it also acts as the identity when $u\in M$. 
Therefore $D \ket{\phi}=\ket{\phi}$.
Observe also that by construction $\swap\ket{\phi}=-\ket{\phi}$, since $f_{u,v}=-f_{v,u}$.
Hence, we conclude that  $U(R,M)\ket{\phi}=\ket{\phi}$.
\end{proof}

The final missing piece is the following property of QPD which will allow us to discriminate $1$-eigenvectors
from those corresponding to eigenvalues with a large eigenvalue gap from $1$.
\begin{lemma}[Quantum Phase Detection]\label{qpd}
Let $U$ be some unitary with given controlled access to it,
and $\ket{\phi}$ be an eigenvector of $U$ with eigenvalue $e^{i\alpha}$, with $\alpha\in(-\pi,\pi]$.
Algorithm QPD$(U,T)$ on input state $\ket{\phi}$ outputs $1$ with probability $0$ if $\alpha=0$,
and otherwise with probability
$$p=\frac{1}{2}-\frac{1}{2T}\cdot\frac{\sin\!\big(\tfrac{T\alpha}{2}\big)\,\cos\!\big(\tfrac{(T+1)\alpha}{2}\big)}{\sin\!\big(\tfrac{\alpha}{2}\big)}.
$$
In particular, 
$p> 2/5$ when $T\geq 20/|\alpha|$.
\end{lemma}
\begin{proof}
For a fixed $t$, the final state is
$$\frac{1}{2}\left((1+e^{it\alpha})\ket{0}+(1-e^{it\alpha})\ket{1}\right)\ket{\phi}.
$$
The probability to output $1$ is therefore 
$p(t)=\sin(t\alpha/2)^2$, which is always $0$ when $\alpha=0$.

We continue assuming that $\alpha\neq 0$. Then,
$$p=\E_{t\in\set{1,2,\ldots,T}} \sin^2(t\alpha/2)=\frac{1}{2}-\frac{1}{2}\E_{t\in\set{1,2,\ldots,T}} \cos (t\alpha),$$
which leads (via Lagrange's trigonometric identities) to the next and second part of the lemma statement: $p=\frac{1}{2}-\frac{1}{2T}\cdot\frac{\sin\!\big(\tfrac{T\alpha}{2}\big)\,\cos\!\big(\tfrac{(T+1)\alpha}{2}\big)}{\sin\!\big(\tfrac{\alpha}{2}\big)}
$.

We now proceed to the last part of the lemma statement. Since this formula does not depend on the sign of $\alpha$, 
and using the assumption that $|\alpha|\geq 4/T$, we can assume that $\alpha\in [4/T,\pi]$. 
Using $\sin x\geq (2/\pi) x$ when $0\leq x \leq\pi/2$, we get 
$$\sin(\alpha/2)\geq (2/\pi)(\alpha/2)\geq \alpha/\pi.$$
We then upper bound the other trigonometric terms by $1$, leading to
$$p\geq \frac{1}{2} -\frac{1}{2T} \times \frac{\pi}{\alpha} \geq \frac{1}{2}-\frac{\pi}{40}
> \frac{2}{5}.$$
\end{proof}

\section{Communication-Optimal MST and Leader Election}
\label{sec:LE}

In the classical setting, any distributed algorithm solving leader election or minimum spanning tree (MST) construction must send $\Omega(m)$ message~\cite{Kutten_2015_JACM}. In this section, we show that there exists a major quantum communication advantage for these two fundamental distributed computing problems. More precisely, $\tilde{O}(n)$ quantum messages suffices to solve leader election and MST, and thus for dense communication graphs we can obtain a quadratic quantum communication advantage. Moreover, these quantum communication upper bounds are tight (up to polylogarithmic factors) due to the lower bound shown in \Cref{sec:lb}.

We start by showing, in \Cref{subsec:clusterPrimitive}, how to leverage distributed quantum walks based on electric networks (whose framework is described in \Cref{sec:distributedQuantumWalk}) to solve a fundamental building block at the heart of leader election --- outgoing edge finding --- using only $\tilde{O}(n)$ quantum communication. Then, in \Cref{subsec:LEMST}, we leverage this more communication-efficient primitive to obtain distributed algorithms for MST and leader election that use only $\tilde{O}(n)$ quantum communication. 

\subsection{Quantum Outgoing Edge Finding}
\label{subsec:clusterPrimitive}

Recall that $G$ denotes the communication graph.
We assume that the edges of $G$ are weighted and that each edge's weight is known to its two endpoints. (If no non-trivial edge weight function is given, we simply assume that all edges have weight 1.) We also assume that $G$ is clustered into some node-disjoint subsets $V_1,\ldots,V_k$, and that each subset $V_i$ is spanned by some tree $T_i$, such that nodes know which incident edge is part of $T_i$, and whether the edge leads to a parent or child in the tree. 

In such a setting, a fundamental distributed primitive --- which we call \findany{} --- is to compute, for each cluster, an \emph{outgoing edge} (i.e., leading to some node outside the cluster) if one exists. We provide a quantum, communication-optimal version of \findany{}. Moreover, this primitive can be easily modified to find an outgoing edge within a certain edge weight range. Thus, via a simple binary search and at the cost of an $O(\log n)$ blowup in message and round complexities, we can also obtain a quantum, communication-optimal version for the weighted variant --- which we call \findmin{} --- that consists in computing, for each cluster, the \emph{minimum weight outgoing edge} if one exists.

\paragraph{Preliminaries.} 
For any cluster $V_i$, we assume that each node in $V_i$ knows the cluster size $|V_i|$, from now on denoted by $n_i$, since $n_i$ can be obtained via a simple convergecast and broadcast over $T_i$. Additionally, each node knows some estimate $n^*$ of the network size, which we use to abort the primitive for large clusters (say, cluster $V_i$ with $n_i > n^*$) for whom the primitive would take too many rounds to successfully terminate. Finally, recall that each node knows a polynomial upper bound $N$ on $n$. Thus, they can choose a large enough ID of $O(\log n)$ bits uniformly at random, such that these IDs are unique with high probability.

Furthermore, for any cluster $V_i$, we call a \emph{good node} any node in $V_i$ that is incident to some outgoing edge. Since at most $n_i$ edges incident to a given node in $V_i$ may lead back to $V_i$, any good node need only send (up to) $n_i+1$ messages to find an incident outgoing edge. Note that as such, any \emph{high degree node} in $V_i$ --- i.e., whose degree is strictly greater than $n_i$ --- is good. 

\paragraph{Description of \findany{}.}
We describe our quantum \findany{} primitive for a single cluster $V_i$ with spanning tree $T_i$. 
The primitive takes $O(n^* \log^{5/2} n)$ rounds, separated into two phases. To start with, if the estimate $n^*$ is bad --- that is, if $n^* < n_i$ --- then the cluster $V_i$ does nothing for these two phases (but nodes in $V_i$ may reply to messages from nodes outside $V_i$). Otherwise, the first phase deals with the easy case --- when the cluster $V_i$ contains a high degree node --- using the first $O(n^*)$ rounds. After which, the second phase takes the remainder of the rounds to deal with the other, harder case --- searching through the at most $n_i^2$ edges incident to $V_i$ --- using only $O(n_i)$ messages, by leveraging distributed quantum walks.

\textbf{In the first phase}, the cluster $V_i$ searches for a high degree (and thus good) node. More precisely, if any such high degree node exists, then the root finds and informs the one with the minimum ID; the high degree node in turn finds an outgoing edge using few messages. Otherwise, if no such high degree node exists, then all nodes in $V_i$ are aware upon termination of this first phase. 

This is done as follows. First, the root broadcast its ID--- which we call the cluster's ID --- over $T_i$. (Doing so also informs each node of its depth in $T_i$, which we use during the second phase.) After which, nodes in $V_i$ perform a convergecast over $T_i$, for which the high degree nodes' messages are their own ID whereas the other nodes' messages are empty. During this convergecast, any node in $T_i$ waits until it receives one message from each of its children. After receiving all such messages, any non-root node sends the \emph{aggregate message} --- that is, the minimum ID among all received IDs, and otherwise the empty message --- to its parent, whereas the root node broadcasts the aggregate message over $T_i$. Doing so informs the high degree node with minimum ID, if it exists. In which case, that high degree node in turn sends a message containing its cluster's ID over $n_i+1$ of its incident edges (and the other endpoints of these edges reply with their cluster's ID). Since only $n_i$ of these edges may lead back to $V_i$, that high degree node finds at least one outgoing edge within these $n_i+1$ edges. 
 
\textbf{In the second phase}, if no high degree node was found previously, the cluster $V_i$ runs $O(\log n)$ distributed quantum walks (one after another) to find a (low degree) good node and compute an outgoing edge, if any exists. (If a high degree node was found in the first phase, then nodes in $V_i$ do nothing except reply to messages coming from outside $V_i$.) In more detail, each distributed quantum walk allows $T_i$'s root to detect w.h.p. whether there exists a good node (i.e., incident to some outgoing edge) whose ID lies in some given range of the ID space. By executing a binary search over the (polynomial in $n$) ID space, the root finds w.h.p. the good node with minimum ID in $V_i$. That node in turn sends a message containing its cluster's ID over its (at most $n_i$) incident edges, and if it is good, at least one of these must be an outgoing edge.

It remains to describe the distributed quantum walk for cluster $V_i$ and for a given range $\mathcal{I}$ of the ID space --- see Section \ref{subsec:distQuantumWalkImplementation} for the generic framework --- but note that distributed quantum walks from different clusters might execute concurrently. (In fact, they might even walk on the same node, but subject to the constraints from Corollary \ref{cdqrw-cor2}.) To do so, we define a implicit electric network (corresponding to $V_i$) that defines the distributed quantum walk --- the corresponding token message contains the cluster's ID (i.e., the randomly chosen ID of that cluster's root) as well as the range $\mathcal{I}$ (i.e., the first and last ID).

In more detail, this implicit electric network is a weighted version of the subgraph $G_i = (U_i, E_i)$, where if we let $V_i(\mathcal{I})$ be the nodes in $V_i$ whose ID is within $\mathcal{I}$, then $U_i = V_i \cup N(V_i(\mathcal{I}))$ and $E_i$ consists of the edges of $T_i$ together with all edges incident to $V_i(\mathcal{I})$. As for the weights, they are defined as follows: for any edge in $T_i$, the weight is the highest depth (in $T_i$) between the edge's two endpoints, whereas any other edge of $E_i$ has weight 1. Moreover, the electric network is rooted in $T_i$'s root, and the set of marked elements $M_i$ consists of all nodes in $N(V_i(\mathcal{I})) \setminus V_i$.
Note that, as required by the distributed quantum walk primitive, each node with the token can compute, using the token's content, whether they are in $M_i$ or not, as well as, when they are unmarked, which of their incident edges are in $G_i$ as well as the corresponding edge weights.

\paragraph{Analysis of \findany{}.} We start with some guarantees for the two phases of our quantum \findany{} primitive --- see \Cref{lem:firstPart,lem:secondPart}.

\begin{lemma} 
\label{lem:firstPart}
    The first phase of \findany{} uses $O(n)$ messages. Moreover, for any cluster $V_i$ such that $n_i \leq n^*$, the cluster finds at least one outgoing edge if and only if there exists a high degree node in that cluster.
\end{lemma}

\begin{proof}
    Consider any cluster $V_i$ satisfying $n_i \leq n^*$. Then, altogether, its nodes send $O(n_i)$ messages and take at most $O(n_i)= O(n^*)$ rounds to execute the convergecast, broadcast and have the high degree node with minimum ID (if one exists) find an incident outgoing edge. 
    On the other hand, consider any cluster $V_j$ satisfying $n_j > n^*$. Then, nodes in $V_j$ send no messages except for replying to messages sent from outside $V_j$, thus from some cluster $V_i$ satisfying $n_i \leq n^*$. Then, we can simply allocate the cost of any such reply message to the sender in $V_i$, and double its message cost. Hence, letting $I \subseteq [k]$ denote all indices corresponding to some cluster $V_i$ satisfying $n_i \leq n^*$, it follows that the message complexity of the first phase is at most $O(\sum_{i \in I} n_i) = O(n)$.
\end{proof}

\begin{lemma}
\label{lem:secondPart}
    The second phase of \findany{} uses $O(n \log^{5/2} n)$ messages. Moreover, it holds with high probability that, for any cluster $V_i$  such that $n_i \leq n^*$ and which contains no high degree node, the cluster finds at least one outgoing edge if and only if there exists a good node in that cluster.
\end{lemma}

\begin{proof}
Consider any cluster $V_j$ satisfying $n_j > n^*$. Then, nodes in $V_j$ send no messages except for replies (to messages from outside $V_j$). Similarly, since any cluster $V_i$ satisfying $n_i \leq n^*$ and containing some high degree node finds an outgoing edge during the first phase by Observation \ref{lem:firstPart}, such a cluster sends no messages during the second phase, except for replies, by the algorithm description. For these two types of clusters, as in the proof of Observation \ref{lem:firstPart}, we allocate the message cost of such replies to the sender node, thus doubling its message cost.
Hence, it suffices to consider the clusters in $I \subseteq [k]$, where $I$ denote all indices corresponding to some cluster $V_i$ satisfying $n_i \leq n^*$ and that contains no high degree node. 

We show next that if the clusters in $I$ each execute a single distributed quantum walk, concurrently, they send at most $O(n \log^{3/2} n)$ messages overall, and with high probability each such cluster finds an outgoing edge in some given range of the ID space, if one exists. 
To do so, we first analyze the electric network defining the distributed quantum walk corresponding to some $V_i$ for $i \in I$. The total weight of that network is at most $W_i=2n_i^2$ since
$\sum_{e\in E_i} w_e \leq  n_i \times\mathrm{depth}(T_i) + n_i^2 \leq W_i$. Next, assume the set of marked nodes (for $V_i$), which we denote by $M_i$, is non empty, and consider any marked node $v \in M_i$. Then, we can define a flow $f$ on the electric network, which is non-zero on all edges except on a single path from $T_i$'s root $r$ to $v$: that path consists of a shortest path within $T$ from $r$ to a neighbor of $v$ in $V_i$, and an additional edge from that neighbor to $v$. The energy of such a flow is:
$$R(f) = \sum_{e\in E_i} \frac{f_e^2}{w_e} \leq 1 + \sum_{i=1}^{\mathrm{depth}(T_i)} \frac{1}{i} \leq 1+\log (\mathrm{depth}(T_i))\leq 1+\log n_i.$$

Thus, the effective resistance of this electric network is at most $R_i = 1+\log n_i$, and we have $\sqrt{R_i W_i} = O( \sqrt{\log n_i} \cdot n_i)$. 
Additionally, any two different electric networks (corresponding to two different clusters) are edge-disjoint except on edges ${u,v}$ where at least one of the node is marked for one of the two clusters, and the various tokens from different clusters are pairwise distinct (as they contain different cluster IDs). Hence, from Corollary~\ref{cdqrw-cor2}, it holds with high probability that the roots of the clusters each detect whether there exists a good node whose ID is within the given range --- this range can differ between clusters --- using in total $O(\log n \cdot \sum_{i=1}^k \sqrt{R_iW_i}) = O(n \log^{3/2} n)$ messages and $O(\log n \cdot \max\limits_{i=1,\ldots,k} \sqrt{R_iW_i}) = O(n^* \cdot  \log^{3/2} n)$ rounds. 

Finally, we take into account the binary search over the ID range, and the cost of finding an outgoing message once a good node with minimum ID is found. First, since the ID range is polynomial in $n$, clusters in $I$ execute at most $O(\log n)$ distributed quantum walks. This takes at most $O(n^* \cdot  \log^{5/2} n)$ rounds, during which they send at most $O(n \log^{5/2} n)$ messages. Moreover, this ensures that with high probability (due to the error probability of the distributed quantum walk), each cluster in $I$ finds a good node with the minimum ID within that cluster if and only if the cluster contains a good node. After which, each such cluster takes an extra $O(1)$ rounds, concurrently, and sends at most $O(n_i)$ messages, to find an incident outgoing edge. Overall, this adds up to an extra $O(n)$ messages, and the lemma statement follows.
\end{proof}

Combined, Observation \ref{lem:firstPart} and Lemma \ref{lem:secondPart} imply the following theorem.

\begin{theorem}
\label{thm:findany}
    \findany{} takes $O(n^* \log^{5/2} n)$ rounds, sends $O(n \log^{5/2} n)$ messages and guarantees with high probability that each cluster $V_i$ such that $|V_i| \leq n^*$ finds an outgoing edge (if any exists).
\end{theorem}

\paragraph{Extension to \findmin{}.} To extend to the weighted case, we first modify the \findany{} primitive and ensure that with high probability, for any cluster $V_i$ and any given edge weight range, cluster $V_i$ computes an outgoing edge within that edge range, if any exists. This can be done by having each node of $V_i$ ignore, throughout the \findany{} primitive, any incident edge whose weight is not in the given range nor is part of $T_i$. 

Given such a modified \findany{} primitive, it is straightforward to see that a simple binary search on the edge weights allows each cluster $V_i$ to compute the minimum weight outgoing edge, if any exists, albeit at a $O(\log n)$ blowup in both round and message complexities. We capture this in the following statement.

\begin{theorem}
\label{thm:findmin}
    \findmin{} takes $\tilde{O}(n^*)$ rounds, sends $\tilde{O}(n)$ messages and guarantees with high probability that each cluster $V_i$ such that $|V_i| \leq n^*$ finds the minimum weight outgoing edge (if any exists).
\end{theorem}

\subsection{Quantum Distributed Algorithm for Minimum Spanning Tree}
\label{subsec:LEMST}

We describe a quantum distributed MST algorithm running in $\tilde{O}(n)$ rounds and messages. This algorithm is a synchronous version of the well-known Gallager-Humbert-Spira (GHS) distributed algorithm for constructing an MST \cite{GHS}, within which we plug in our communication-optimal quantum outgoing edge finding primitive from \Cref{subsec:clusterPrimitive}. 
Naturally, this algorithm also solves the leader election, spanning tree construction, and broadcast problem in $\tilde{O}(n)$ rounds and messages.

\paragraph{Algorithm Description.}
The algorithm works in phases with exponentially increasing estimates of the network size: that is, for phase $i \geq 1$, the estimate is $n^* = 2^i$. The algorithm terminates upon detecting that a single tree spans the whole network. As shown in the analysis, this takes at most $O(\log n)$ phases. (Alternatively, one could run the algorithm for a fixed $O(\log n)$ phases, determined by the analysis, at the cost of an additional inversely polynomial in $n$ error probability for the correctness.)

Each phase is further subdivided into $O(\log n)$ subphases. At the start of each phase, each node is its own \emph{MST fragment} (i.e., a subtree of the MST). Then, during each subphase, we merge MST fragments together by executing the following three steps:
\begin{enumerate}
    \item First, we run \findmin{} (i.e., the quantum minimum weight outgoing edge finding primitive from \Cref{subsec:clusterPrimitive}) with the estimate $n^*$ so that with high probability, any MST fragment with size smaller than or equal to $n^*$ finds its minimum weight outgoing edge. As a result, we can construct a virtual fragment (super)graph whose (super)nodes are the current MST fragments and where there exists an edge between two MST fragments if there exists a minimum weight outgoing edge with an endpoint in each of the two fragments. If one directs an edge within the fragment graph from the cluster having found that edge outwards, then it is known that (due to the properties of minimum weight outgoing edges) each connected component is spanned by two directed trees rooted at two neighboring fragments $F_1$ and $F_2$, where $F_1$ and $F_2$ both chose the same minimum weight outgoing edge (along opposite directions).
    \item Second, each fragment $F_i$ with size smaller than or equal to $n^*$ sends the ID of its root over the minimum weight outgoing edge it found in the previous step --- this is done by broadcasting over the fragment's tree to have the endpoint within $F_i$ of that minimum weight outgoing edge send a message over that edge. If a message comes along that same edge in the opposite direction --- i.e., that endpoint also receives an ID along the same edge --- then that ID is forwarded to the fragment's root. In which case, if the fragment ID is higher, the fragment becomes the root of the connected component.  
    \item Third, we merge each connected component in this virtual graph into a single MST fragment --- this can be done via a well-known distributed computing primitive. In short, the component's root fragment informs its neighboring fragments (via a broadcast over the fragment, and sending messages over incident minimum weight outgoing edges). These neighboring fragments redirect their tree edges, but also inform their neighboring fragments, and so on. To successfully complete, this ``naive'' merge operation may take up to $O(n_t)$ messages and rounds, where $n_t$ is the sum of the size of the fragments within that connected component. Here, to finish (whether successfully or not) the merge within $O(n^*)$ rounds, we abort the merge upon reaching some $O(n^*)$ rounds, and then
    we use an extra $O(n^*)$ rounds to reverse the merge and keep the fragments separate.
\end{enumerate}
After some $O(\log n)$ subphases, each resulting cluster executes the quantum outgoing edge finding primitive again. If the cluster finds an outgoing edge or if its size is strictly greater than $n^*$, then it goes to the next phase, otherwise it terminates.

\paragraph{Analysis.} We start with two basic lemmas -- see \Cref{lem:MSTFragment,lem:conditionalProgress}. 

\begin{lemma}
\label{lem:MSTFragment}
    For any phase and subphase of the MST algorithm, each fragment is a subtree of the MST with high probability.
\end{lemma}

\begin{proof}
    Consider any given phase. Then, the lemma statement follows by induction on the subphases. Indeed, the base case is trivially true. As for the induction step, suppose we start with MST fragments. Then, \findmin{} ensures that with high probability, each MST fragment (with size smaller than $n^*$) compute its minimum weight outgoing edge if and only if one exists. After which, we either successfully merge along these edges, or reverse the merges. It is a well-known fact that by merging along minimum weight outgoing edges, we maintain MST fragments and the induction step follows (with high probability).
\end{proof}

\begin{lemma}
\label{lem:conditionalProgress}
    Consider any phase with $n^* \geq n$. Then, for any of its subphase, the number of fragments reduces by half with high probability. Moreover, the merging takes $O(n)$ rounds.
\end{lemma}

\begin{proof}
    Consider any such subphase $i \geq 1$. Let $\mathcal{F}_i$ be the set of fragments at the start of the subphase. Then, by definition of $n^*$ and Theorem \ref{thm:findmin}, it holds with high probability that each fragment in $\mathcal{F}_i$ finds a minimum weight outgoing edge. Thus, there are at most $\mathcal{F}_i/2$ connected components in the virtual fragment supergraph. Each such component contains at most $n$ nodes, and thus $O(n^*)$ rounds suffice to successfully merge together the fragments making up that component. In other words, each such component corresponds to exactly one fragment when the subphase ends, and the lemma statement follows.
\end{proof}

Using these lemmas, we can prove this section's main result.

\begin{theorem}
\label{thm:MST}
    There exists a quantum distributed algorithm that takes $\tilde{O}(n)$ messages and rounds to solve MST w.h.p.
\end{theorem}

\begin{proof}
    We first prove termination as well as the complexity upper bounds. To start with, each subphase takes $O(2^i \log^{7/2} n)$ rounds and messages, and thus each phase takes at most $O(2^i \log^{9/2} n)$ rounds and messages. Then, suppose that the algorithm has not terminated prior to phase $i^*=\lceil \log n \rceil$. In which case, during phase $i^*$, the estimate is good: that is, $n^* \geq n$. Then, by Lemma \ref{lem:conditionalProgress}, there can be at most $O(\log n)$ subphases before a single fragment remains, with high probability. By Theorem \ref{thm:findmin}, the fragment finds no outgoing edge. Moreover, $n^*$ is greater than the fragment size, thus all nodes in the fragment terminate in phase $i^*$. Adding up rounds and messages up until that point, we get that the round and message complexities are $O(\sum_{i=1}^{i^*} 2^i \log^{9/2} n) = O(n \log^{9/2} n)$. 

    Finally, we show correctness. By Lemma \ref{lem:MSTFragment}, each fragment is a subtree of the MST with high probability. Moreover, by Theorem \ref{thm:findmin}, any fragment with size smaller than $n^*$ that does not contain all nodes find an outgoing edge with high probability. On the other hand, fragments with a larger size go on to the next phase. In summary, a union bound over the fragments and the $O(\log^2 n)$ subphases (over all phases) implies that with high probability, nodes terminate if and only if their fragment spans the whole graph, or in other words, if and only if the fragment is the MST. 
\end{proof}

\section{Communication-Optimal Breadth First Search}
\label{sec:BFS}

Similarly to leader election and minimum spanning tree, any classical distributed algorithm computing a BFS must send $\Omega(m)$ messages \cite{Kutten_2015_JACM}. In this section, we show that in the quantum setting however, $\tilde{O}(\sqrt{mn})$ messages suffice for BFS. In other words, we show that there exists a major quantum communication advantage for this fundamental distributed computing problem. Moreover, this quantum communication advantage is tight (up to polylogarithmic factors) due to the lower bound shown in \Cref{sec:lb}.

The quantum BFS algorithm we provide is inspired by the classical algorithm from \cite{DPRS25}. First, in \Cref{subsec:groverSearch}, we describe the distributed Grover search framework. Then, in \Cref{subsec:primitivesBFS}, we show how quantum communication can be leveraged to solve low-depth BFS explorations and sparse neighborhood cover computations, in a message-efficient manner. Finally, in \Cref{subsec:BFSAlg}, we use sparse neighborhood covers and distributed Grover searches to obtain our communication-optimal quantum distributed algorithm for BFS.

\subsection{Distributed Grover Search}
\label{subsec:groverSearch}

We describe a quantum subroutine --- distributed Grover search --- run by some node $u$ in $G$. Let $X$ be a finite set, and $f\colon X\to \{0,1\}$ a function that indicates which elements of $X$ are \emph{marked}. Let $t_f=|f^{-1}(1)|$ and $\eps_f=t_f/|X|$ be, respectively, the number and fraction of marked elements in $X$. Then, the subroutine allows node $u$ to find a marked element $x\in X$ (i.e., such that $f(x)=1$), if any such element exists. 

Moreover, let there be a distributed algorithm $\checking$ that enables $u$ to check whether $x$ is marked, or in other words, to compute $f(x)$ for an input $x$ known by $u$. 
(For synchronization constraints, we require an upper bound on the number of rounds used by $\checking$, for any input $x\in X$.)
More formally, let $\ket{\psi}_G$  be the initial state of the network, which may result from some preliminary initialization steps, and we assume without loss of generality that it is a quantum state.
Given $x\in X$ in (the local memory of) $u$, $\checking$ enables $u$ to compute $f(x)$ in time $T_\Cc$ with message complexity $M_\Cc$ as follows, where the subscript $u$ means that the register is local to $u$, and $G$ that it is global to the network:
$$\checking: \ket{x,0}_u \ket{\psi}_G\mapsto \ket{x,f(x)}_u\ket{\phi_x}_G.$$

Then, a classical version of the subroutine can use the following strategy:
node $u$ samples $x\in X$ and then checks whether $f(x)=1$ using $\checking$. In which case, the success probability is $\eps_f$, which can be boosted to $(1-\alpha)$ with $\Theta( \log(\tfrac{1}{\alpha})/{\eps_f})$ iterations. 
When $\eps_f$ is not known but satisfies either $\eps_f=0$ or $\eps_f\geq \eps$, where $0<\eps\leq 1$ is some input parameter, then one can guarantee that after $\Theta({\log(\frac{1}{\alpha})}/{\eps})$ iterations (where each iteration runs $\checking$ once, thus takes $T_\Cc$ rounds and $M_\Cc$ messages), node $u$ can distinguish between the two cases with arbitrarily high success probability.

But quantumly, we can do quadratically better (in $\eps$). The following is a distributed adaptation of the Grover Search algorithm~\cite{Grover96}, in the case of an unknown number of pre-images of $f$~\cite[Lemma 2]{bbht98}, to the distributed setting. 
Such an adaptation was firstly done in~\cite{LM18} for the round complexity.
Here, we also consider the message complexity as in~\cite{DMP25}.

\begin{theorem}[Distributed Grover Search~\cite{DMP25}]\label{dgs}
Let $f,u,\eps_f$, $\checking$, $T_\Cc$ and $M_\Cc$ be defined as above.
Then, for any $\eps,\alpha>0$, there is a quantum distributed algorithm $\mathsf{GroverSearch}(\eps,\alpha)$
such that
\begin{enumerate}
\item $\mathsf{GroverSearch}(\eps,\alpha)$ runs in $O({\log(\frac{1}{\alpha})}\times\frac{T_\Cc}{\sqrt{\eps}})$ rounds with message complexity $O({\log(\frac{1}{\alpha})}\times\frac{M_\Cc}{\sqrt{\eps}})$;
\item $\mathsf{GroverSearch}(\eps,\alpha)$ returns to
$u$ some $x\in X$, which satisfies $f(x)=1$ with probability at least $1-\alpha$ when $\eps_f\geq \eps$.
\end{enumerate}
\end{theorem}

\subsection{Quantum Low-Depth BFS Explorations and Sparse Neighborhood Covers}
\label{subsec:primitivesBFS}

First, we show how quantum communication enables message-efficient distributed low-depth BFS explorations. 
More precisely, let $S \subseteq V$ be some source nodes, $d \geq 1$ be the exploration depth, and $k$ be the congestion, such that for any node $v \in V$, there are at most $k$ BFS roots within $d$ hops of $v$.\footnote{We assume throughout this section that all nodes know $n$, which they use to synchronize their Grover searches. This assumption can be lifted by using the algorithm from \Cref{sec:LE} at the cost of an extra $\tilde{O}(n)$ rounds and messages; indeed, given a spanning tree, a simple convergecast and broadcast suffice to compute the network size $n$.} Then, to solve the distributed low-depth BFS exploration task consists in computing, for each node $s \in S$, a BFS tree up to depth $d$ rooted in $s$. Upon termination, each node should know which such BFS tree they are part of, and if so, which incident edges are part of that tree and whether they lead to their parent or children (within that tree).

We present a quantum distributed algorithm for low-depth BFS explorations working in phases. These phases each grow all BFS trees by exactly one layer, but we do not have nodes in the ``outermost'' layer, or \emph{frontier}, of the BFS tree contact the nodes that make up the next layer. In other words, we do not grow the BFS tree ``inside out''. Instead, we grow the BFS tree ``outside in'', that is, we have the outside nodes search for any neighbor in the BFS tree and upon finding one, join the next layer. Intuitively, the ``inside out'' approach may lead to multiple nodes within a BFS contacting the same outside node, and thus to a communication inefficiency. On the other hand, the ``outside in'' approach enables us to grow the layers fast (i.e., in a single Grover search) while also using only $\tilde{O}(k\sqrt{n})$ messages per layer. 
 
The algorithm runs for $d$ phases; these phases are synchronized across all nodes, and take $\tilde{\Theta}(k\sqrt{n})$ rounds each.
Each such phase $1 \leq i \leq d$ is further subdivided into $k$ subphases $\tilde{\Theta}(\sqrt{n})$ rounds each. During each subphase $1 \leq j \leq k$ (of some phase $i$), each node $v$ searches its neighborhood (using a distributed Grover search) for any node that is part of a BFS tree $v$ has not joined prior to the subphase. If $v$ finds any such node $u$, then it joins the corresponding BFS tree with $u$ as its parent. (Note that $v$ can be part of multiple BFS trees.)

For more details regarding the distributed Grover search, $v$ runs $\mathsf{GroverSearch}$ with $\eps=1/ \deg(v)$, $\alpha = \Theta(1/n^c)$ for some constant $c$, and the function $f_v: N(V) \rightarrow \{0,1\}$ that assigns $f(w) = 1$ to any node that is in some BFS tree which $v$ has not joined prior to the subphase, and $f(w) = 0$ to any other node in $N(v)$. Note that the corresponding checking function is trivial, with $v$ sending a message to $w$ and $w$ replying, and uses $2$ rounds and messages. 

\begin{theorem}
\label{thm:SmallBFSExploration}
    Let integers $d, k$ be respectively the distance and congestion parameters, and $S \subseteq V$ the BFS roots, such that for any node $v \in V$, there are at most $k$ BFS roots within $d$ hops of $v$. 
    Then, there exists a quantum distributed algorithm that (with high probability) computes $d$-depth BFS trees rooted in all nodes of $S$, in $\tilde{O}(d k \cdot \sqrt{n})$ rounds and $\tilde{O}(dk \cdot \sqrt{mn})$ messages.
\end{theorem}

\begin{proof}
We prove the correctness by induction on phase $i \in \{0,\ldots,d\}$ (where the end of phase $0$ simply means the initial state) that by the end of phase $i$, all nodes that are some $j \leq i$ hops away from some BFS root $r \in S$ have joined the BFS tree of $r$ in layer $j$. The base case is trivial, thus we now consider that the induction hypothesis holds for some $i \in \{0,\ldots,d-1\}$. Consider phase $i+1$ as well as some node $v \in V$ such that $v$ is exactly $i+1$ hops away from some BFS root $r \in S$. Then, $v$ contains at least one neighbor exactly $i$ hops away from $r$, but no neighbor strictly less than $i$ hops away from $r$. By the induction hypothesis, by the end of phase $i$, any neighbor of $v$ exactly $i$ hops away from $r$ have joined (layer $i$ of) the BFS rooted in $r$. Hence, by Theorem \ref{dgs} and the fact that there are at most $k$ BFS roots within $d$ hops of $v$ (by definition of $k$ and $d$), with high probability, $v$ finds such a neighbor within one of the $k$ subphases of phase $i+1$ and joins the BFS tree rooted in $r$ in layer $i+1$. The induction step follows.

As for the round complexity, it follows from the algorithm's description. Finally, we prove the message complexity upper bound. Note that for each subphase, each node $v \in V$ sends $O(\sqrt{\deg(v)} \cdot \log n)$ messages. Hence, each subphase sends overall $O(\sum_{v \in V} \sqrt{\deg(v)} \cdot \log n) = O(\sqrt{m n} \log n)$ messages (by the Cauchy-Schwarz inequality). Since there are at most $dk$ subphases, the claimed message complexity upper bound follows.
\end{proof}

The above quantum distributed primitive for low-depth BFS explorations uses only $o(m)$ quantum communication for a wide range of distance and congestion parameters. However, setting $d=D$ and $S = \{r\}$ to obtain a BFS rooted in some node $r$ leads to a non-optimal message complexity of $\tilde{O}(D \cdot \sqrt{mn})$. 
Instead, our communication-optimal BFS algorithm (see \Cref{subsec:BFSAlg}) hinges on the following primitive: sparse neighborhood cover construction. 
We start by defining sparse neighborhood covers below.  

\begin{definition}\label{def:cover}
A \emph{sparse $(\kappa,W)$-neighborhood cover} of a graph is a collection $\cC$ of trees,
each called a {\em cluster}, with the following properties.
\begin{itemize}
\item \emph{(Depth property)} For each tree $\tau \in \cC$, depth$(\tau) = O(W \cdot \kappa)$.
\item \emph{(Sparsity property)} Each vertex $v$ of the graph appears in $\tilde O(\kappa \cdot n^{1/\kappa})$ different trees $\tau \in \cC$.
\item \emph{(Neighborhood property)} For each vertex $v$  of the graph there exists a tree $\tau \in \cC$ that contains the entire $W$-neighborhood of vertex $v$.
\end{itemize}
\end{definition}

Next, we briefly describe the distributed (randomized) neighborhood cover construction due to Elkin~\cite{Elkin06}. Although this algorithm uses $\Omega(m)$ messages, we show below --- after the classical description --- how it can be made more message-efficient in the quantum setting.
The algorithm runs in $\kappa$ phases, each of $O(\kappa \cdot n^{1/\kappa} \log n \cdot \kappa W)$ rounds. Initially, nodes are uncovered. Then, in each phase $i \in [1,\kappa]$, a well-chosen number of uncovered vertices initiate a BFS exploration of depth $2((\kappa-i)+1) W$, each such exploration forming a cluster of the cover. (The property of the BFS exploration primitive ensures the depth property of the cover, while the judicious choice of the number of sources per phase ensures the sparsity property---see Lemma~\ref{lem:smallMaxDegreeOfCover} below.) Nodes that join such a cluster, within at most $2(\kappa-i) W$ hops of the root, become covered; indeed, their $W$-hop neighborhood is contained in that cluster, and hence in the cover, thus satisfying its neighborhood property.

We obtain a quantum distributed algorithm for constructing sparse neighborhood covers --- which we call \quantumCover{} --- by modifying the above (classical) algorithm as follows: we replace the classical (but message-inefficient in the quantum setting) BFS exploration primitive with a quantum (message-efficient) BFS exploration primitive. It follows that most properties of distributed (randomized) cover construction due to Elkin~\cite{Elkin06} directly translate over to the modified version. In particular, the correctness of \quantumCover{} follows directly (but with an additive polynomially small failure probability). Additionally, it also holds for \quantumCover{} that with high probability, the maximum degree of the constructed cover is small (see Lemma \ref{lem:smallMaxDegreeOfCover}); in other words, each vertex appears in a small number of clusters of the cover, in fact in $O(\kappa \cdot n^{1/\kappa} \log n)$ of them. 

\begin{lemma}[Corollary A.6 in \cite{Elkin06}]
\label{lem:smallMaxDegreeOfCover}
    The sparse neighborhood cover $\cC$ computed by \quantumCover{} satisfies the following: with high probability, for every vertex $v \in V$,
    $|\{C \in \cC \mid v \in C\}| = O(\kappa \cdot n^{1/\kappa} \log n)$.
\end{lemma}

\begin{lemma}
    There exists a quantum distributed algorithm that constructs (with high probability) a $(\kappa, W)$-sparse neighborhood cover in time $O(\kappa^2 W n^{1/\kappa+1/2})$ and using $O(\kappa^2 W n^{1/\kappa} \sqrt{mn})$ messages.
\end{lemma}

\begin{proof}
    Note that BFS explorations are done up to only $d= 2\kappa W$ hops in \quantumCover{}. Moreover, by Lemma \ref{lem:smallMaxDegreeOfCover}, each node is contained in at most $k = O(\kappa \cdot n^{1/\kappa} \log n)$ clusters of the cover (over all phases). 
    Hence, by Theorem \ref{thm:SmallBFSExploration}, the round and message complexities of \quantumCover{} are respectively $\tilde{O}(\kappa^2 W n^{1/\kappa+1/2})$ and $\tilde{O}(\kappa^2 W n^{1/\kappa} \sqrt{mn})$.
\end{proof}

\subsection{Quantum Distributed Algorithm for BFS Tree Computation}
\label{subsec:BFSAlg}

Finally, we give our communication-optimal quantum distributed algorithm for BFS. For that algorithm, similarly to the BFS explorations in \Cref{subsec:primitivesBFS}, we grow the BFS tree via the ``outside in'' approach.
However, a major difference with the ``naive'' BFS explorations from \Cref{subsec:primitivesBFS} lies in our use of a sparse neighborhood cover. Using the cover's tree, we can inform nodes as to whether they are close to the frontier, and only at this point do these nodes execute a distributed Grover search. This is crucial for achieving $\tilde{O}(\sqrt{mn})$ message complexity.

We now describe our quantum BFS algorithm, with  some root node $r \in V$ as input.
To start with, as a preprocessing step, we compute a $(\Theta(\log n), 1)$-sparse neighborhood cover.
After which, the algorithm runs for $O(D)$ phases, where each phase $i \leq D$ growing the BFS tree by exactly one layer. (Nothing happens during phases $i > D$, but the algorithm may not have terminated yet.) 
Each phase is further separated into two subphases, the first taking $O(\log^3 n)$ rounds and the second $O(\sqrt{n} \log n)$ rounds. Moreover, every phase $i=2^j$ for some integer $j \geq 1$ contains an additional subphase, which takes $O(2^j)$ rounds. During that subphase, the BFS tree convergecasts to detect whether any node joined in the previous phase, followed by a broadcast over the tree to inform all nodes of the detection's outcome. If no node joined during the previous phase, then all nodes terminate upon receiving that broadcast. 

\textbf{In the first subphase}, nodes compute whether they are close --- i.e., at most $O(\log n)$ hops away --- to the frontier of the BFS tree
via the sparse neighborhood covers. More precisely, any frontier node --- i.e., having joined the BFS tree in the previous phase --- \emph{pings} any tree of the cover it belongs to --- i.e., the node sends a message to its parent, which in turn sends the message to its parent, and so on until it reaches the root. Any pinged cover tree, in turn, broadcasts a message to inform all nodes within that tree that they are close to the frontier. (Note that any node not in the BFS but with a neighbor in the frontier must get informed, due to the neighborhood property of the sparse cover.)

\textbf{In the second subphase}, any informed node $v$ that has not already joined the BFS runs a distributed Grover search on its neighborhood --- just as in the low-depth BFS explorations --- to find whether it has any neighbor in the frontier (i.e, the ``outermost'' layer of the BFS tree). If it does, it joins the BFS with that node as its parent and will be part of the frontier for the next phase. For more details regarding the distributed Grover search, $v$ runs $\mathsf{GroverSearch}$ with $\eps=1/ \deg(v)$, $\alpha = \Theta(1/n^c)$ for some constant $c$ and the function $f_v: N(V) \rightarrow \{0,1\}$ that assigns $f(w) = 1$ to any node that is in the BFS tree. (Once again, the corresponding checking function is trivial, and uses $O(1)$ rounds and messages.) 

\begin{theorem}
    There exists a quantum distributed algorithm that takes $\tilde{O}(\sqrt{mn})$ messages to compute a BFS tree w.h.p. Its round complexity is upper bounded by $\tilde{O}(D\sqrt{n})$ when $n$ is known, and $\tilde{O}(D\sqrt{n}+n)$ otherwise. 
\end{theorem}

\begin{proof}
    We prove the correctness by induction on phase $i \in \{0,\ldots,D\}$ (where the end of phase $0$ simply means the initial state) that by the end of phase $i$, all nodes at distance $j \leq i$ from the BFS root $r$ has joined layer $j$ of the BFS tree. The base case is trivial, thus we now consider that the induction hypothesis holds for some $i \in \{0,\ldots,d-1\}$. Consider phase $i+1$ as well as some node $v \in V$ such that $v$ is exactly $i+1$ hops away from $r$. Then, $v$ contains at least one neighbor exactly $i$ hops away from $r$, but no neighbor strictly less than $i$ hops away from $r$. By the induction hypothesis, by the end of phase $i$, any neighbor of $v$ exactly $i$ hops away from $r$ have joined (layer $i$ of) the BFS tree. Consider any such neighbor $u$. Then, in phase $i+1$, $u$ pings the roots of any of the cover's trees that node $u$ belongs to, and these roots in turn inform all nodes in their tree. By the depth and sparsity properties for $\kappa = \Theta(\log n)$ and $W = 1$, the $O(\log^3 n)$ rounds suffice for the ping and subsequent broadcast to successfully terminate (despite the CONGEST communication). Moreover, by the neighborhood property, the neighborhood of $v$ is entirely contained within one tree of the cover, and thus $v$ is informed following $u$'s ping during (the first subphase of) phase $i+1$. Following which, $v$ runs a distributed Grover search. By Theorem \ref{dgs}, with high probability, $v$ finds $u$ or another neighbor in the BFS tree and joins the BFS tree in layer $i+1$. The induction step follows. 

    As for the round complexity, first note that the algorithm runs for at most $O(D)$ phases. Indeed, there exists a phase $i \leq 2(D+1)$ with a third subphase during which the convergecast detects that no node joined the BFS tree in the previous phase. Additionally, each phase takes $\tilde{O}(\sqrt{n})$ rounds by the algorithm description, except for phases $i=2^j$ for some integer $j \geq 1$, which take an extra $O(2^j)$ rounds. Adding up the runtime of all these phases together, we obtain a runtime of $\tilde{O}(D \sqrt{n})$, assuming that $n$ is known. Otherwise, if $n$ is not initially known, then computing a spanning tree as well as its size $n$ takes an extra $\tilde{O}(n)$ rounds, by Theorem \ref{thm:MST}. 
    
    Finally, we prove the message complexity upper bound. First, note that each node $v$ executes at most $O(\log n)$ distributed Grover searches. Indeed, for $v$ to execute a distributed Grover search in some phase $i$ , there must exist a cover tree and a node within that cover tree that joined the BFS tree in phase $i-1$. Since the cover's trees all have depth $O(\log n)$, the nodes within all cover trees that $v$ belongs to are within $O(\log n)$ layers of $v$ in the BFS. Consequently, since each distributed Grover search run by node $v \in V$ sends $O(\sqrt{\deg(v)} \cdot \log n)$ messages, by Theorem \ref{dgs}, then overall, nodes send $O(\sum_{v \in V} \sqrt{\deg(v)} \cdot \log^2 n) = O(\sqrt{m n} \log^2 n)$ messages (by the Cauchy-Schwarz inequality) over all distributed Grover searches. It remains to bound the communication induced by the convergecast and broadcast over the BFS tree, as well as the communication over the cover trees. For the first, note that the convergecast and broadcast over the BFS tree happen at most $O(\log D) = O(\log n)$ times, thus the corresponding communication is at most $O(n \log n)$. As for the communication over the cover trees, note that the pings and subsequent broadcasts can only happen during $O(\log n)$ different phases per cover tree, since nodes within a given cover tree may join the BFS tree within at most $O(\log n)$ different phases. As a result, the communication over the cover trees amounts to $O(n \log n)$ messages in total. Added together, we obtain the claimed message complexity upper bound.
\end{proof}

\section{Quantum Message Lower Bounds}
\label{sec:lb}

We lift quantum query lower bounds on (unweighted) $n$-vertex graphs --- in the centralized setting --- to lower bounds in our quantum distributed setting. By doing so, we show the first non-trivial communication lower bounds (see Theorems \ref{lb:BFS} and \ref{lb:LE}) for quantum distributed computing.\footnote{Note that in~\cite{legall2025}, only a classical communication lower bound is provided, using a different reduction and other techniques.} More concretely, we show an $\Omega(\sqrt{mn})$ quantum communication lower bound for BFS and Single Source Shortest Paths (SSSP), as well as an $\Omega(n)$ quantum communication lower bound for (implicit) leader election, ST and MST. 

In \Cref{subsec:LBTechniques}, we provide some preliminaries, including a straightforward reduction from query complexity (in the sequential model) to quantum communication complexity (in the distributed model). Then, in \Cref{subsec:BFSLB}, we show an $\Omega(\sqrt{mn})$ query complexity lower bound for BFS, and leveraging that, $\Omega(\sqrt{mn})$ quantum communication lower bounds for BFS and SSSP. Finally, in \Cref{subsec:LELB}, we show an $\Omega(n)$ query complexity lower bound for connectivity in low diameter graphs (i.e., graphs of diameter 3 and above), and leveraging that, $\Omega(n)$ quantum communication lower bounds for (implicit) leader election, ST and MST.

\subsection{Model and Lower Bound Techniques}
\label{subsec:LBTechniques}

The distributed setting considered throughout this section is defined in \Cref{sec:qmodel}. For the centralized setting, we consider \new{a slight variation of} the adjacency array model~\cite{DHHM06}, in which the vertices' degrees $d_1,\ldots,d_n$ are initially known and for every vertex $v \in \{1,\ldots,n\}$, we can query its neighbors through an array $f_v : [d_v] \rightarrow [n]$.
We assume that the arrays represent a simple, undirected and unweighted graph. In this adjacency array model, the complexity measure of interest is the quantum query complexity: that is, the number of quantum queries made over the adjacency arrays. Then, formally, a \emph{query} is a pair $i\new{=} (v,p)$ where $1\leq p\leq d_v$.
We then say that the query $i$ corresponds to the (oriented) edge $(v,f_v(p))$.
\new{The \emph{answer} is then not only $u=f_v(p)$, but $(u,q)$ where $q$ encodes the rank of $v$ in $f_u$, \emph{i.e.} $f_u(q)=v$. This additional request was not in the original adjacency array model~\cite{DHHM06}, but it is required for our reduction lemma below. We interpret ranks $p,q$ as ports of $v$ and $u$ in the distributed setting.}

To show lower bounds on this quantum query complexity, we use the adversary method with nonnegative weights~\cite{Ambainis02}, captured by the following theorem,
where by bounded error we mean an error smaller than some arbitrary constant, usually $1/3$. In fact, we consider a slight extension of the original theorem, which was implicit~\cite{Ambainis02} and made explicit in several works, see for instance in~\cite{LM08}.
\new{We will use this theorem with $x,y$ encoding input graphs, $i\new{=} (v,p)$ a query, and
$x_i=(u,q)$ the corresponding answer. Then the condition $x_i \neq y_i$ means that the query $i$ distinguishes between the two input graphs.}
\begin{theorem}[\cite{Ambainis02}]
\label{lowerbound}
Let $F(x_1, \ldots, x_\n)$ be a function of $\n$ variables 
with values from some finite set and $S$ be a set of inputs. 
Let $X,Y\subseteq S$ and $R\subseteq X \times Y$ be such that 
\begin{enumerate}
\item For every $(x,y)\in R$, we have $F(x)\neq F(y)$;
\item For every $x\in X$, there exist at least $\m$ different $y\in Y$ such that
$(x, y)\in R$;
\item
For every $y\in Y$, there exist at least $\mm$ different $x\in X$ such that
$(x, y)\in R$.
\end{enumerate}
Let $l_{x, i}$ be the number of $y\in Y$ such that $(x, y)\in R$ and $x_i\neq y_i$
and $l_{y, i}$ be the number of $x\in X$ such that $(x, y)\in R$ and $x_i\neq y_i$.
Let $\lmax$ be the maximum of $l_{x, i}l_{y, i}$ over all $(x, y)\in R$
and $i\in\{1, \ldots, \n\}$ such that $x_i\neq y_i$.
Then, any quantum algorithm computing $F$ on $S$ with bounded error 
uses  $\Omega(\sqrt{\frac{\m \mm}{\lmax}})$ queries.
\end{theorem}

We also use a variant of this theorem for relations, which can be easily derived  from~\cite{Ambainis02,LM08}. In order to keep the same formalism, we then denote by $F(x_1,\ldots,x_{\n})$ the set of valid outputs, formally in relation to the input $(x_1,\ldots,x_{\n})$. In that case, we say that a quantum algorithm computing $F$ has bounded error $\delta$ when (the measure of) its outputs is one of the valid outputs with probability at least $1-\delta$. 
\begin{theorem}
\label{lowerbound2}
Let $F(x_1, \ldots, x_\n)$ be a function of $\n$ variables 
with values being subsets of some finite set, and $S$ be a set of inputs.
Let $X,Y\subseteq S$ and $R\subseteq X \times Y$ be such that
\begin{enumerate}
\item For every $(x,y)\in R$, we have $F(x)\cap  F(y)=\emptyset$;
\item
For every $x\in X$, there exist at least $\m$ different $y\in Y$ such that
$(x, y)\in R$;
\item
For every $y\in Y$, there exist at least $\mm$ different $x\in X$ such that
$(x, y)\in R$.
\end{enumerate}
Let $l_{x, i}$ be the number of $y\in Y$ such that $(x, y)\in R$ and $x_i\neq y_i$
and $l_{y, i}$ be the number of $x\in X$ such that $(x, y)\in R$ and $x_i\neq y_i$.
Let $\lmax$ be the maximum of $l_{x, i}l_{y, i}$ over all $(x, y)\in R$
and $i\in\{1, \ldots, \n\}$ such that $x_i\neq y_i$.
Then, any quantum algorithm computing $F$ on $S$ with bounded error for each $z\in S$ uses $\Omega(\sqrt{\frac{\m \mm}{\lmax}})$ queries.
\end{theorem}

As mentioned previously, we can lift lower bounds from the centralized setting to the distributed setting: more precisely, we can lift quantum query complexity lower bounds into quantum communication lower bounds via a simple reduction. This is captured by the following statement.

\begin{lemma}[Query Complexity to Message Complexity (QCMC) Reduction Lemma]
\label{lem:CentralizedToDistributedLowerBound}
    For any quantum distributed algorithm solving some (unweighted) distributed graph problem $\mathcal{P}$ on some communication graph $G$ with communication cost $M$, there exists a quantum algorithm which, given query access to the adjacency array of $G$, uses $M$ queries to solve $\mathcal{P}$ on $G$.
\end{lemma}

\begin{proof}
    The quantum algorithm in the adjacency array model straightforwardly simulates the quantum distributed algorithm. Clearly, it can simulate the memory and local computations of the $n$ vertices in a centralized manner (since we do not bound the memory and computation cost in the centralized setting), but it can also simulate the quantum communication within the distributed algorithm. Indeed, it suffices to execute one quantum query whenever a quantum message is sent in the distributed algorithm. \new{More precisely, the communication within each round of the distributed algorithm gets simulated as follows. First, for each message sent (from some node $v$ over some port $p$ to some neighbor $u$) by the algorithm, a query $i=(v,p)$ is executed by the centralized model, and this \emph{modified query} returns $(u,q)$. Second, the corresponding communication from $v$ to $u$ is simulated by a $\swap$ operator: that is, the message is moved from $v$ to $u$ by exchanging the (simulated) $p$th emission register of $v$ and the (simulated) $q$th reception register of $u$.} (Note that quantum queries are more general since they are not necessarily limited to edges incident to one node, unlike quantum messages.)
    
    In summary, the above described centralized quantum algorithm solves $\mathcal{P}$ since the distributed algorithm does, whereas its quantum query complexity follows from the one-to-one correspondence between quantum messages and quantum queries within the above reduction.
\end{proof}

Note that the above reduction is limited to unweighted distributed graph problems. Indeed, for weighted graphs, the adjacency array model assumes that no edge weight is known initially, and to learn an edge's weight, one must query the corresponding edge. On the other hand, the distributed setting assumes that the two endpoint nodes of any given edge know its weight initially. This explains, in particular, why the MST problem requires $\Omega(\sqrt{mn})$ queries within the adjacency array model \cite{DHHM06} but can be solved using only $\tilde{O}(n)$ messages within our distributed setting (see \Cref{sec:LE}). 

\subsection{BFS Lower Bound}
\label{subsec:BFSLB}

\begin{theorem}
\label{lb:centralizedBFS}
For any positive integers $n$ and $d\leq n-1$, there exists a diameter-$2$ graph $G$ with $\Theta(n)$ vertices and $m = \Theta(nd)$ edges such that given a query access to its adjacency array, the quantum query complexity of computing the BFS of $G$ from a given root node $r$ is $\Omega(n\sqrt{d}) = \Omega(\sqrt{mn})$.
\end{theorem}

\begin{proof}
We prepare our setting for applying Theorem~\ref{lowerbound2}. The function $F$ we consider is defined informally as the set of encodings of any valid BFS from some given vertex $r$, where each edge $\{u,v\}$ of that BFS is encoded by one of its endpoints, say $u$, and the corresponding port, say $p$, such that $f_u(p)=v$.

Before constructing our relation set $R$, we first describe the set $S$ of graphs we will restrict to. (In the following, we assume $n$ is even. If $n$ is odd, just replace $n$ by $n+1$.) We consider graphs $G$ with $3$ levels, illustrated in \Cref{fig:bfs}, defined as follows:
\begin{enumerate}
    \item Root $r$;
    \item $n$ nodes $a_1,\ldots,a_{n}\in A$ connected to the root;
    \item $d$ nodes $b_1,\ldots,b_d\in B$ connected by a complete bipartite graph to the nodes at level $1$, together with $n$ additional nodes $c_1,\ldots,c_{n}\in C$ connected to the nodes at level $1$ by a matching $M$,
    and connected together by $d$ edge-disjoint matchings.
\end{enumerate}

\begin{figure}[h]
 \centering
\begin{tikzpicture}[every node/.style={font=\small},root/.style={circle, draw=black, fill=yellow!30, minimum size=8mm},blockA/.style={circle, draw=black, fill=blue!15, minimum size=7mm},blockB/.style={circle, draw=black, fill=cyan!20, minimum size=6.5mm},blockC/.style={circle, draw=black, fill=orange!20, minimum size=6.5mm},labelstyle/.style={font=\footnotesize}]
\def\n{6}\def\d{3}\def\spacing{1.4}
\node[root] (r) at (0,2.2) {$r$};
\foreach \i in {1,...,\n}{\coordinate (tmp\i) at ({(\i-1-(\n-1)/2)*\spacing},0); \node[blockA] (l\i) at (tmp\i) {$a_{\i}$}; \draw[thick](r)--(l\i);}
\foreach \i in {1,...,\d}{\pgfmathsetmacro{\xshift}{(\i-1-(\n-1)/2)*\spacing - 3} \node[blockB] (b\i) at (\xshift,-2) {$b_{\i}$};}
\foreach \i in {1,...,\n}{\pgfmathsetmacro{\xshift}{(\i-1-(\n-1)/2)*\spacing + 3} \node[blockC] (c\i) at (\xshift,-2) {$c_{\i}$};}
\foreach \i in {1,...,\n}{\foreach \j in {1,...,\d}{\draw[blue!70,semithick,opacity=0.8](l\i)--(b\j);}}
\foreach \i in {1,...,\n}{\draw[red,very thick](l\i)--(c\i);}
\pgfmathtruncatemacro{\halfn}{\n/2}       
\pgfmathtruncatemacro{\lastshift}{\d}   
\foreach \shift in {0,...,\lastshift}{
  \foreach \i in {1,...,\halfn}{
    \pgfmathtruncatemacro{\a}{1+mod(\i-1+\shift,\n)}
    \pgfmathtruncatemacro{\b}{1+mod(\n-\i+\shift,\n)}
     \ifnum\a>\b
      \pgfmathtruncatemacro{\temp}{\a}
      \pgfmathtruncatemacro{\a}{\b}
      \pgfmathtruncatemacro{\b}{\temp}
    \fi
    \pgfmathtruncatemacro{\bend}{min(45,8*abs(\b-\a))}
    \draw[green!60!black,thin,opacity=0.8] (c\a) to[bend left=\bend] (c\b);
  }
}
\node[labelstyle] at (-4.5,-2.8){Block $B$ (complete bipartite graph with $A$)}; \node[labelstyle] at (4.5,-2.8){Block $C$ ($1$ matching $M$ with $A$, $d$ matchings in $C$)};
\begin{scope}[xshift=3.5cm,yshift=2.0cm]\draw[blue!70,semithick](0,0)--+(0.6,0) node[right,labelstyle]{Complete bipartite graph edges};\draw[red,very thick](0,-0.5)--+(0.6,0) node[right,labelstyle]{Matching $M$};\draw[green!60!black,thin,opacity=0.8](0,-1.0)..controls(0.3,-0.7)..(0.6,-1.0) node[right,labelstyle]{Matchings in $C$};\draw[thick](0,-1.5)--+(0.6,0) node[right,labelstyle]{Root $\to$ level~1 edges};\end{scope}
\end{tikzpicture}
\caption{BFS hard instance with $n=6$ and $d=3$\label{fig:bfs}}
\end{figure}
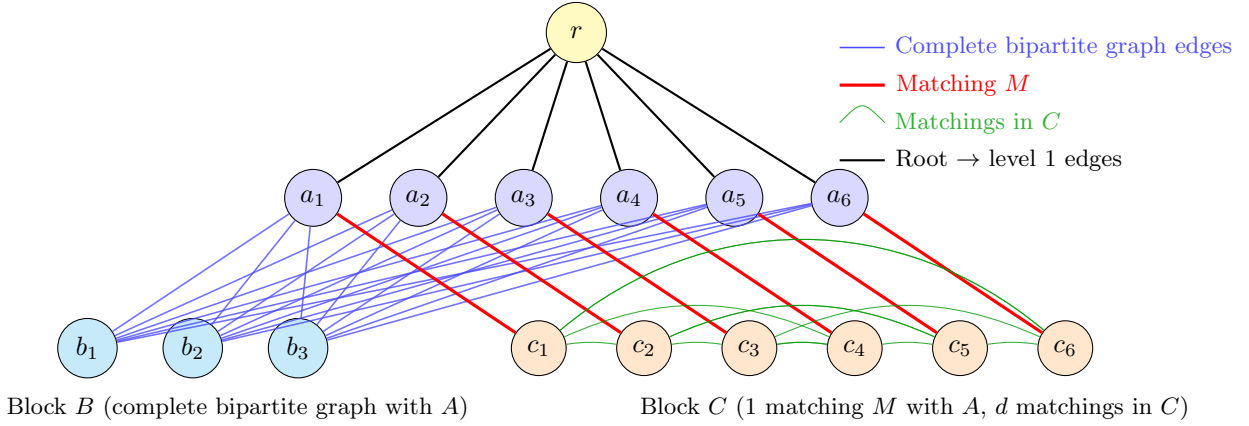

First, note that for any $d \leq n-1$, since $n$ is even, it is indeed possible to connect the $n$ nodes of $C$ together by $d$ edge-disjoint matchings. Moreover, observe that the edges of $M$ must be in the BFS of $G$.
Thus, for each edge $\{a,c\}\in M$, either node $a$ or node $c$ must identify its port corresponding to the edge $\{a,c\}$. 

Now that we have defined $S$, let $X=Y=S$  and the relation set $R\subset S\times S$ be all pairs of such graphs $(G,G')$ such that one can go from $G$ to $G'$ (and conversely) as follows:
\begin{enumerate}
    \item\label{step1} Select two edges $\{a,b\}$, $\{a,c\}$: Permute respectively the ports $p,q$ for $b$ and $c$ in $a$, that is, exchange $f_a(p)=b$ with $f_a(q)=c$; 
    \item\label{step2} Select one edge $\{c,c'\}$: Permute similarly the ports for $a$ and $c'$ in $c$.
\end{enumerate}
\new{More precisely, let $p',q'$ be such that $f_b(p')=a$ and $f_c(q')=a$. Then after Step~\ref{step1}, the answers $(b,p')$ and $(c,q')$ to the queries $(a,p)$ and $(a,q)$ are now exchanged. In return, the answer $(a,p)$ to the query $(b,p')$ is modified to $(a,q)$. Similarly, the answer $(a,q)$ to the query $(c,q')$ becomes $(a,p)$.
Step~\ref{step2} should be expanded similarly.}
Clearly we have $\m=\mm= \Theta(n d^2)$. 

Now let us consider such a pair $(G,G')$. The two graphs differs in 
%$2$ queries in $a$ and $2$ queries in $c$. 
\new{$2$ queries in $a$, $2$ queries in $c$, $1$ query in $b$ and $1$ query in $c'$. Moreover, each one of these queries identifies one of the $3$ edges $\{a,b\}$, $\{a,c\}$ and $\{c,c'\}$.}
We are going to show that $\lmax\leq d^3$, therefore we will get a query complexity lower bound of $\Omega(\sqrt{\frac{\m\mm}{\lmax}}) = \Omega(\sqrt{\frac{n^2 d^4}{d^3}}) = \Omega(n \sqrt{d})$.

Indeed, let us fix one such query $i$ among those \new{$6$} %$4$ 
possible queries. If $i$ corresponds to the edge $(a,b)$ (meaning that $i$ corresponds to querying the port of $a$ that leads to $b$\new{, or conversely}), then the remaining number of possibilities to modify $G$ (for the choice of \new{$c'$, since $c$ is uniquely defined by $a$}) is $d$.
The same argument holds if $i$ corresponds to $(c,c')$. 
The hardest cases are when $i$ corresponds to some edge $(a,c)$, or respectively $(c,a)$. Indeed, in that case, the number of possibilities for the pair $(b,c')$ is $d^2$. But then we observe that, in the corresponding graph $G'$, the query $i$ corresponds to the edge $(b,a)$, or respectively $(c,c')$, for which the number of sibling graphs is then only $d$, as shown in the previous two cases. Hence, $\lmax\leq d^3$.
\end{proof}

\begin{theorem}
\label{lb:BFS}
    Any quantum distributed algorithm solving Breadth First Search must send $\Omega(\sqrt{mn})$ quantum messages.
\end{theorem}

\begin{proof}
    Assume by contradiction that there exists some quantum distributed BFS algorithm using $o(\sqrt{mn})$ messages. Then, by Lemma \ref{lem:CentralizedToDistributedLowerBound}, there exists a quantum (centralized) algorithm, which, given any input root node $r$ and some query access to some graph $G$, computes a BFS of $G$ rooted in $r$ using $o(\sqrt{mn})$ queries. In particular, this holds for the graph considered in Theorem \ref{lb:centralizedBFS}, which leads to a contradiction. As a result, we obtain the desired distributed quantum communication lower bound for BFS.
\end{proof}
 
A simple corollary of this quantum communication lower bound --- via a trivial reduction in the distributed setting --- is that the Single Source Shortest Paths (SSSP) problem also admits an $\Omega(\sqrt{mn})$ quantum communication lower bound. 
Combined with our BFS algorithm from \Cref{sec:BFS}, we obtain that this communication lower bound is tight for BFS; however, whether the lower bound is also tight for SSSP remains open.

\subsection{Implicit Leader Election Lower Bound}
\label{subsec:LELB}

Deciding connectivity requires $\Theta(n)$ queries on $n$-vertex graphs in the adjacency array model~\cite{DHHM06}. However, this query lower bound weakens to $\Omega(d)$ for diameter-$d$ graphs. 
Since we aim to lift a quantum query complexity lower bound and show a distributed lower bound for graphs of diameter 3 (and higher), we adapt the connectivity lower bound proof from~\cite{DHHM06} and show $\Omega(n)$ queries are required to distinguish between connected diameter-$3$ graphs and graphs made of two disjoint cliques.

\begin{theorem}
\label{lb:centralizedConnectivity}
For any positive integer $n$, there exists a $2n$-vertex $(n-1)$-regular graph $G$, such that given a query access to its adjacency array, 
the quantum query complexity of deciding whether $G$ is connected with diameter $3$, or is the union of two disjoint cliques, is $\Omega(n)$.
\end{theorem}
\begin{proof}

We prepare our setting for applying Theorem~\ref{lowerbound}. The function $F$ we consider is defined by $F(G)=1$ if $G$ is connected and $F(G)=0$ otherwise. We are going to construct $R$ such that $R=  \{G\}\times Y$, thus $X=\{G\}$.
We fix a graph $G$ and its adjacency arrays $f_v:[n]\to[2n]$, for $v\in[2n]$ as follows. Let $G$ be the union of two disjoint $n$-cliques, the first one over vertex set $[n]$, and the second one over $(n,2n]$. Then $(f_v)_{v\in[2n]}$ represents any fixed encoding of $G$. Thus $\mm=l_{y,\new{j}}=1$.

We now consider $\m=n^4$ possible modifications of $G$ (and its adjacency arrays) in order to generate the set $Y$. 
The graphs encoded in $Y$, and illustrated in \Cref{figspan}, are obtained from $G$ with two additional bridge connecting the two cliques, 
and the removal of one edge in each clique in order to preserve the regularity of the resulting graph.
More precisely, for $(a,b)\in[n]^2$ and $(c,d)\in (n,2n]^2$, we construct a graph $G^{(a,b,c,d)}\in Y$
such that edges $\{a,b\}$ and $\{c,d\}$ are deleted from $G$ and replaced by edges $\{a,c\}$ and $\{b,d\}$. 

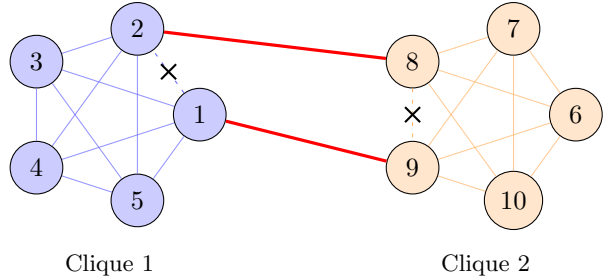
\begin{figure}[h]
\centering
\begin{tikzpicture}[every node/.style={font=\small},
clique1/.style={circle,draw=black,fill=blue!20,minimum size=7mm},
clique2/.style={circle,draw=black,fill=orange!20,minimum size=7mm},
labelstyle/.style={font=\footnotesize}]
\def\n{5}\def\radius{1.2}
\foreach \i in {1,...,\n}{\pgfmathsetmacro{\angle}{360/\n*(\i-1)}\coordinate (L\i) at ({\radius*cos(\angle)},{\radius*sin(\angle)});\node[clique1] (v\i) at (L\i) {$\i$};}
\foreach \i in {1,...,\n}{\pgfmathsetmacro{\angle}{360/\n*(\i-1)}\coordinate (R\i) at ({\radius*cos(\angle)+5},{\radius*sin(\angle)});\node[clique2] (w\i) at (R\i) {$\the\numexpr \i+\n\relax$};}
\foreach \i in {1,...,\n}{\foreach \j in {\i,...,\n}{\ifnum\i<\j \ifnum\i=1 \ifnum\j=2\else \draw[blue!60, thin, opacity=0.7] (v\i)--(v\j); \fi \else \draw[blue!60, thin, opacity=0.7] (v\i)--(v\j); \fi \fi}}
\foreach \i in {1,...,\n}{\foreach \j in {\i,...,\n}{\ifnum\i<\j \ifnum\i=3 \ifnum\j=4\else \draw[orange!60, thin, opacity=0.7] (w\i)--(w\j); \fi \else \draw[orange!60, thin, opacity=0.7] (w\i)--(w\j); \fi \fi}}
\draw[blue!60, thin, dash pattern=on 0.5mm off 1mm] (v1)--(v2);
\draw[orange!60, thin, dash pattern=on 0.5mm off 1mm] (w3)--(w4);
\path (v1) -- (v2) coordinate[midway] (mv);
\draw[black, thick] (mv) ++(-0.1,-0.1) -- ++(0.2,0.2);
\draw[black, thick] (mv) ++(-0.1,0.1) -- ++(0.2,-0.2);
\path (w3) -- (w4) coordinate[midway] (mw);
\draw[black, thick] (mw) ++(-0.1,-0.1) -- ++(0.2,0.2);
\draw[black, thick] (mw) ++(-0.1,0.1) -- ++(0.2,-0.2);
\draw[red,very thick] (v2)--(w3);
\draw[red,very thick] (v1)--(w4);
\node[labelstyle] at (0,-2) {Clique 1};
\node[labelstyle] at (5,-2) {Clique 2};
\end{tikzpicture}
\caption{Spanning-tree hard instance with $n=5$\label{figspan}}
\end{figure}

The corresponding encoding $G^{(a,b,c,d)}$ is obtained by modifying the adjacency arrays $(f_v)_{v\in[2n]}$ of $G$ in $4$ entries. Let 
%$i,j,k,l\in [n]$ be such that $f_{a}(i)=b$, $f_{b}(j)=a$, $f_{c}(k)=d$ and $f_{d}(l)=c$. Then the new lists are given by $f^{(a,b,c,d)}_{a}(i)=c$, $f^{(a,b,c,d)}_{b}(j)=d$, $f^{(a,b,c,d)}_{c}(k)=a$ and $f^{(a,b,c,d)}_{d}(l)=b$.
\new{$p,q,p',q'\in [n]$ be such that $f_{a}(p)=b$, $f_{b}(q)=a$, $f_{c}(p')=d$ and $f_{d}(q')=c$. Then the new lists are given by $f^{(a,b,c,d)}_{a}(p)=c$, $f^{(a,b,c,d)}_{b}(q)=d$, $f^{(a,b,c,d)}_{c}(p')=a$ and $f^{(a,b,c,d)}_{d}(q')=b$.}

We then directly observe that $l_{x,i}\leq n^2$.
%(\new{with $x$ being $G$ and $i$ one of the potential $4$ above queries)}.
Therefore, the query complexity for deciding whether $F(G) = 1$ or $F(G)=0$ (i.e., whether $G$ is connected or not) requires $\Omega(n)$ queries to $G$.
\end{proof}

\begin{theorem}
\label{lb:LE}
    Any quantum distributed algorithm solving Implicit Leader Election on graphs of diameter at least 3 must send $\Omega(n)$ quantum messages.
\end{theorem}

\begin{proof}
    Assume by contradiction that there exists some quantum distributed algorithm solving implicit leader election on graphs of diameter at least 3 in $o(n)$ messages. Note that this also includes disconnected communication graphs, and for such graphs, each connected component essentially runs one independent instance of the leader election algorithm. Hence, the above distributed algorithm
    would also guarantee that exactly one leader is elected per connected component.
    Thus, by Lemma \ref{lem:CentralizedToDistributedLowerBound}, there exists a quantum algorithm which, given query access to any graph $G$ of diameter at least 3, uses only $o(n)$ queries to compute exactly one leader per connected component of $G$. By subsequently counting the number of leaders (which requires no extra query, only a simple computation), this quantum algorithm can distinguish between connected diameter-3 graphs and graphs made of two disjoint cliques in $o(n)$ queries. This contradicts Theorem \ref{lb:centralizedConnectivity}, and thus proves the desired distributed quantum communication lower bound for implicit leader election.
\end{proof}

In fact, the above $\Omega(n)$ quantum communication lower bound applies to any global problem that can be reduced to leader election using sublinear (in $n$) communication in our distributed setting. In particular, the above lower bound also applies to spanning tree (ST) and minimum spanning tree (MST) construction; the corresponding reduction is straightforward: after computing a spanning tree, the root becomes the leader. Combined with our MST algorithm from Section \ref{sec:LE}, we get that this $\Omega(n)$ communication lower bound is tight for these problems. However, note that when the communication graph has diameter 1 or 2, there exist algorithms using $o(n)$ quantum communication~\cite{DMP25}.

\section{Conclusion and Open Problems}
\label{sec:conclusion}

We presented distributed quantum algorithms for various fundamental problems that are nearly optimal in terms of their message complexity in the quantum routing model. We also proved their near optimality by presenting almost matching quantum message lower bounds. 
Our techniques — both algorithmically and lower bound-wise — can be helpful for designing communication-optimal distributed algorithms for other significant problems, such as shortest paths and minimum cuts.

Our work raises important open questions. Perhaps the most fundamental is whether we can improve the round complexity of our communication-optimal algorithms. As it is, our algorithms are only existentially optimal with respect to round complexity (cf. Section \ref{sec:contributions}). We {\em conjecture} that a tradeoff exists between communication and round complexity in the quantum setting. In particular, we conjecture that it may not be possible to improve the round complexity of distributed quantum algorithms for the problems considered here if we require them to be communication-optimal (or near-optimal). Proving or disproving this conjecture
and/or showing round-communication tradeoffs will be of great interest.

\bibliographystyle{alpha}
\bibliography{biblio}

\appendix

\section{Discussion on the quantum routing model and its physical feasibility}

\label{app:formal}

Both fields of quantum distributed computing and quantum networks are highly active today. For some people, these may appear to be hypothetical scenarios; however, the early years of quantum computing were similarly characterized by numerous objections from physicists, and today, the first quantum computers are on the horizon.

Building single-quantum computers that operate with a large number of noise-free qubits has been a major challenge. Indeed, till recently, it has been difficult to build quantum computers with more than a few {\em logical} qubits \cite{googlequantum}.
Quantum computers may, in fact, require distributed architectures to scale, since individual quantum units may be limited to a relatively small number of qubits due to physical constraints.(see, e.g.,\cite{jacinto2025,quantumdatacenter}). 
At the same time, large quantum cryptography systems need to rely on large quantum networks, sometimes called {\em quantum Internet} (see e.g.,
\cite{QIA, julien}).

Although the two distributed notions differ, both have motivated the development of the theoretical foundations of quantum distributed computing.
Most of these studies are from the perspective
of {\em round} complexity, showing that either there is no quantum advantage for some problems or there is for some other problems (cf.
Section \ref{subsec:additionalRelatedWork}).

Regarding {\em message} complexity, the work is more recent \cite{icalp24,DMP25,roget2025, legall2025,robinson2026}. This can be understood from several recent developments in sublinear-message-complexity randomized (and thus classical) protocols, as well as from several physics experiments demonstrating that quantum routing is physically achievable.

The scalability of quantum routing in a superposition of trajectories will face physical implementation challenges, which are beyond the scope of this work. However, such challenges are not dissimilar to those encountered when realizing quantum access to classical RAM. Such a component is indeed a necessary condition for quantum advantage in several applications in classical machine learning. 

Indeed, when it comes to communication, the most promising quantum resources are photons. However, with respect to distributed resources, experiments have shown that a single photon in a coherent superposition of $N$ trajectories is easier to realize than $N$ indistinguishable, synchronized photons. Indeed, an experiment has already been reported in 2010 with $4$ paths entangled within $4$ quantum memories \cite{naturephotons}. There have also been recent experiments demonstrating quantum control for routing photons \cite{opticaphoton}. Only recently has the manipulation of synchronized and identical photons become possible.
One of the world's track records is 12 entangled photons obtained by postselecting 6 sources of photon pairs \cite{prlphoton}. 
In fact, manipulating more photons is such a challenge that some quantum technologies combine both approaches in order to increase the number of encoded qubits \cite{meng2024versatile}.

Such quantum networks would require a perfect setting that relies on frequent synchronization over all the network links. But this part is not computing anything, and we do not take it into account in our complexity measure, where we assume
that a single photon can be sent in superposition
to its neighbors with little overhead. This is similar to several quantum analogues of classical complexity measures, such as {\em quantum query complexity} (cf. Section \ref{subsec:additionalRelatedWork}).
Let us also emphasize that, as in the usual LOCAL/CONGEST models, the local topology of each node is known locally; thus, nodes know their incoming and outgoing links, 
possibly together with the physical implementation parameters (such as the travel time per link for practical implementation considerations).

\end{document}